\documentclass[12pt]{article}
\usepackage[margin=1in]{geometry}
\usepackage{longtable}
\usepackage{titlesec}\titleformat{\section}  {\normalfont\Large\bfseries}  {\thesection}{1em}  {}
\titleformat{\subsection}  {\normalfont\large\bfseries}  {\thesubsection}  {1em}
  {}

\usepackage{amssymb}

\newcommand{\R}{\mathbb{R}}
\newcommand{\OO}{\mathcal{O}}

\usepackage{dsfont}
\usepackage{xcolor}
\newcommand{\hlroi}[1]{\textcolor{red}{#1}}

\usepackage{natbib}
\usepackage{tabularx}
\usepackage{longtable}

\usepackage{algorithm, algorithmic}
\usepackage{enumitem}
\usepackage{amsmath}
\usepackage{amsthm}
\usepackage{hyperref}
\hypersetup{colorlinks=true, allcolors=blue}
\usepackage{authblk}
\usepackage{tabularx}
\usepackage{booktabs}
\usepackage{graphicx}
\newtheorem{theorem}{Theorem}
\newtheorem{remark}{Remark}

\usepackage{lineno}

\title{SpARCD: A Spectral Graph Framework for Revealing Differential Functional Connectivity in fMRI Data} 

\author[1]{Shira Yoffe}
\author[3]{Ziv Ben-Zion}
\author[6]{Guy Gurevitch}
\author[4,5,6,7]{Talma Hendler}
\author[2]{Malka Gorfine}
\author[1]{Ariel Jaffe}

\affil[1]{Department of Statistics and Data Science, Hebrew University of Jerusalem, Jerusalem, Israel}

\affil[2]{Department of Statistics and Operations Research, Tel Aviv University, Tel Aviv, Israel}

\affil[3]{School of Public Health, Faculty of Social Welfare and Health Sciences, University of Haifa, Haifa, Israel}

\affil[4]{Sagol School of Neuroscience, Tel Aviv University, Tel Aviv, Israel}

\affil[5]{Sagol Brain Institute, Tel Aviv Sourasky Medical Center, Tel Aviv, Israel}

\affil[6]{Gray Faculty of Medical and Health Sciences, Tel Aviv University, Tel Aviv, Israel}

\affil[7]{School of Psychological Sciences, Faculty of Social Sciences, Tel Aviv University, Tel Aviv, Israel}

\date{}

\begin{document}

\maketitle

\begin{abstract}

Identifying brain regions that exhibit altered functional connectivity between cognitive or emotional states is a fundamental problem in neuroscience. We propose SpARCD (Spectral Analysis for Revealing Connectivity Differences), a statistical framework for detecting 
detecting condition-specific patterns of functional connectivity.
SpARCD uses distance correlation, a dependence measure sensitive to both linear and nonlinear associations, to construct weighted region-wise connectivity graphs for each condition. A differential operator obtained through spectral filtering is then used to identify connectivity changes via its leading eigenvectors.
To assess statistical significance, we develop a permutation-based testing procedure that yields interpretable region-level significance maps. We establish finite-sample validity of the permutation test and derive asymptotic guarantees for the stability of the resulting region rankings. Simulation studies demonstrate improved power relative to conventional edge-wise and univariate approaches, particularly in settings with nonlinear dependence structures.
We applied SpARCD to fMRI data from 113 individuals with early-stage PTSD and 42 controls during emotional and neutral task conditions. The method identified distinct connectivity networks associated with visual processing in both PTSD and control participants. 
Resting-state comparisons between PTSD and control participants highlighted similar visual networks. SpARCD provides a statistically rigorous and computationally efficient framework for comparing high-dimensional connectivity patterns.
 
\end{abstract}











\section{Introduction}

Identifying variables associated with distinct biological or clinical states is a crucial task across many scientific disciplines \citep{smyth2004linear, love2014moderated}. In many applications, differences between states do not arise from the additive effects of a small number of variables, but rather from coordinated changes in groups of interdependent variables \citep{daudt2018urban, xiao2018alternating, zhao2021detection}.
This perspective is particularly relevant in neuroscience, where coordinated activity within brain networks has been linked to behavior and clinical outcomes \citep{smith2009correspondence, crossley2013cognitive}.
Motivated by this principle, in this paper, we develop a framework for analyzing functional magnetic resonance imaging (fMRI) scans to identify brain networks whose connectivity patterns are strongly associated with a medical condition. 

fMRI is a non-invasive neuroimaging technique that indirectly measures brain activity via blood-oxygen-level–dependent (BOLD) signals, reflecting changes in local blood oxygenation that are tightly coupled to neural activity \citep{logothetis2008we}.
A major challenge in fMRI analysis is that regional BOLD signal levels can vary substantially across individuals for reasons unrelated to clinical state.
To mitigate this limitation, researchers often analyze functional connectivity (FC) between brain regions, typically estimated through statistical dependence between their BOLD time series. 
By focusing on coordinated activity rather than absolute regional signal levels, FC provides a network-level representation of brain function.
The robustness of FC has made it a central tool for studying brain networks associated with cognitive, emotional, and clinical differences \citep{greicius2004default, anand2005activity, liao2010altered,mohanty2020rethinking, finn2015functional, van2010intrinsic}. However, reliably detecting connectivity changes remains challenging due to the high dimensionality and complex interdependencies among brain regions \citep{zalesky2010network}.


In this work, we introduce \textit{SpARCD} (Spectral Analysis for Revealing Connectivity Differences), a novel approach for detecting brain networks inspired by recent advances in spectral graph theory. For each condition, our method constructs a graph whose nodes represent brain regions and whose edges encode their functional connectivity. The problem of identifying brain networks that distinguish two conditions is thus formulated as a graph-theoretic task: given two graphs with identical node sets but differing edge weights, detect subsets of nodes whose connectivity profiles differ significantly between the graphs. This spectral formulation provides an efficient, interpretable framework for uncovering network alterations in brain connectivity.

\subsection{Main application} Our primary motivating application is the analysis of fMRI data collected under the Emotional Face Matching Task (EFMT, or Hariri task) \citep{hariri2002amygdala}, a widely used framework for probing emotional reactivity and identifying neural changes associated with emotional processing \citep{savage2024dissecting}. During the task, participants match faces displaying various emotional expressions and geometric shapes that serve as controls, yielding BOLD signals corresponding to emotional and neutral conditions.
This design allows for within-subject comparisons of emotional and neutral states. Prior studies of EFMT revealed FC alterations in conditions such as alcohol use disorder \citep{gorka2013alcohol, o2012withdrawal}, depression \citep{carballedo2011functional}, and schizophrenia \citep{goghari2017task}. Here, we analyze fMRI scans from individuals with PTSD, a condition in which previous work has extensively examined activation patterns in emotion-related brain regions. Investigating dysregulated connectivity patterns among these regions can provide unique insights into the neural basis of PTSD and reveal mechanisms underlying symptom development \citep{shin2010neurocircuitry, etkin2007functional, sripada2012neural}.

We analyze fMRI data from the cohort described in \cite{ben2019neurobehavioral}. The study enrolled 171 adults who underwent clinical and neuroimaging assessments at three post-trauma time points: one month, six months, and fourteen months. Sixteen participants were excluded due to incomplete or poor-quality imaging at the first session, yielding a final cohort of 155 individuals, of whom 113 met full PTSD criteria at the first time-point. Our analysis focuses on resting-state fMRI and the EFMT taken at the first time point, one month post-trauma. This cohort is notable for its large sample size relative to typical neuroimaging studies, the early timing of scanning (one month after trauma), and the deliberate inclusion of participants with high PTSD symptom severity. 
These features provide a unique opportunity to investigate neural correlates of emotion processing and early markers of PTSD shortly after trauma exposure. 

\subsection{Problem Setting and related work} \label{sec: problem setting}


The EFMT consisted of four blocks of facial expressions and five blocks of shapes. For our analysis, we excluded the first block of shapes, corresponding to the initial scanning block, due to its higher noise levels. 

The resulting fMRI data were then parcellated into $R = 113$ regions of interest (ROIs) according to the Harvard–Oxford atlas, comprising 96 bilateral cortical regions and 17 subcortical regions. A detailed mapping of these brain regions is provided in Appendix \ref{app:Atlas mapping}.


Parcellation of the fMRI recordings from the two conditions yields two corresponding datasets: $X \in \R^{n_X \times R \times T},  Y \in \R^{n_Y \times R \times T}$, where $R$ denotes the number of ROIs, $T$ the number of time samples and $n_X,n_Y$ are the number of scans in the two conditions, respectively. The datasets $X$ and $Y$ may be paired, meaning that every scan in $X$ corresponds to a scan in $Y$, or unpaired.  
Our goal is to identify brain networks whose FC patterns differ between the two conditions. To facilitate clarity in the upcoming technical sections, we summarize key notation used throughout the paper in Table~\ref{tab:notation}.

Many studies have analyzed FC between brain regions to gain insight into the organization of brain networks \citep{van2010exploring}.
A common strategy is to test each edge separately, for example, by comparing correlation coefficients between groups and correcting for multiple comparisons \citep{whitfield2012conn}. 
Such procedures, typically based on controlling the False Discovery Rate (FDR), substantially reduce statistical power and may miss effects with weak neuronal signal \citep{nichols2003controlling, eklund2016cluster, zhang2011multiple}. Moreover, univariate tests are inherently limited in their ability to capture coordinated, subnetwork-level alterations \citep{zalesky2010network}.

The Network-Based Statistic (NBS) \citep{zalesky2010network} is a statistical framework designed to address the multiple-comparison problem in edge-wise connectivity analyses by computing a graph whose nodes and edge weights represent brain regions and some measure of statistical association between them.  The method computes a sparse graph by applying a threshold to the graph's connectivity matrix and identifies connected components (subnetworks) within the sparse graph. Finally, a subnetwork-level statistic, such as the subnetwork’s size or summed edge weights, is computed and assessed for significance via permutation testing. By aggregating effects across interconnected edges, NBS increases sensitivity to distributed network-level changes, though its dependence on a fixed primary threshold may reduce robustness and limit detection of small or overlapping effects.

Psychophysiological interaction (PPI) analysis is another widely used approach \citep{friston1997psychophysiological} that has been shown to produce robust and reproducible results \citep{smith2016toward}. 
In contrast to network-level approaches, PPI estimates how task-related changes in neuronal activity modulate the connectivity between a seed region and the rest of the brain. 
However, its main limitation is its reliance on a single a priori seed selection, which can affect the reliability of the results \citep{li2023effect}. In addition, PPI may provide only a partial view of FC, as it does not account for interactions between non-seed regions \citep{gerchen2014analyzing}.

In this work, we introduce a graph-based method for detecting changes in functional connectivity between two fMRI states, which is based on differences in the leading eigenvectors between graphs. Graph-based methodologies were proposed for various tasks in fMRI analysis. For example, small-world network measures have been used to characterize alterations in conditions such as Alzheimer's disease \citep{sanz2010loss}. Within this framework, geodesic-distance-based metrics such as characteristic path length and global efficiency quantify network-level communication and integration. These measures have been shown to be sensitive to disease-related changes in FC \citep{bullmore2009complex, liao2010altered}. 
\citet{jacob2016dependency,jacob2019reappraisal} derived DepNA, which models directed influence relationships among regions, and used it to examine whether connectivity hierarchies
distinguish patients with social anxiety disorder from controls. The work of \citet{ben2017tired} evaluates the relative contribution of each ROI to the separation between two psychological states.
Spectral graph-based approaches have also been applied to fMRI data to uncover functional organization. For instance, spectral clustering has been utilized to create functionally homogeneous parcellations from resting-state fMRI data \citep{craddock2012whole, shen2013groupwise}. Additionally, \citet{cribben2017estimating} applied spectral methods to detect changes in network structure. \citet{cahill2016multiple} utilized spectral methods for group-wise functional community detection.

Here, we focus on spectral graph analysis, which examines the eigenstructure of graph Laplacians to reveal differences in connectivity patterns between states \citep{sristi2022disc, yoffe2024spectral}. 
Our work combines several methodological, theoretical, and empirical components.
First, we develop a spectral graph-based framework designed to capture coordinated shifts in network structure. For each dataset corresponding to a specific state (e.g., a particular stimulus condition), we construct a graph whose edge weights represent estimates of FC between pairs of regions. We estimate FC using distance correlation \citep{Szkely2007distance}, which captures both linear and nonlinear dependencies. This method has proven to provide a more reliable and robust estimate of connectivity, particularly when applied to real data \cite{geerligs2016functional}. Instead of exhaustive pairwise comparisons, we identify meaningful differences in network organization by comparing the leading eigenvectors of the respective Laplacians. Second, we propose a permutation-based inference framework for assessing the statistical significance of these spectral differences, explicitly accounting for data dependencies and yielding interpretable region-level significance scores. On the theoretical side, we establish finite-sample validity of the permutation procedure and derive asymptotic theoretical guarantees for the stability of the eigenvector-based region ranking. Finally, through extensive simulations and an application to fMRI data from PTSD patients engaged in an emotional reactivity task, we demonstrate the utility of the proposed method in uncovering clinically relevant, network-level alterations distinguishing emotional from neutral states.

\begin{table}
\centering
\small
\begin{tabular}{llp{0.58\textwidth}}
\toprule
\textbf{Symbol} & \textbf{Dimensions} & \textbf{Description} \\
\midrule

$X$
& $\mathbb{R}^{n \times R \times T}$
& Tensor containing the signal matrices of $R$ brain regions for $n$ samples across $T$ time points (similarly for $Y$). \\

$X^{(i)}$
& $\mathbb{R}^{R \times T}$
& Signal matrix of the $i$-th sample (similarly $Y^{(i)}$). \\

$X_r$
& $\mathbb{R}^{n \times T}$
& Signal matrix of the $r$-th region across $n$ samples (similarly $Y_r$). \\

$\widehat W_X$
& $\mathbb{R}^{R \times R}$
& Empirical similarity (connectivity) matrix constructed from the observed data in state $X$ (similarly $\widehat W_Y$). \\

$\widehat L_X$
& $\mathbb{R}^{R \times R}$
& Empirical symmetric normalized Laplacian constructed from $X$ (similarly $\widehat L_Y$). \\

$\widehat v_X^{(k)}$
& $\mathbb{R}^{R}$
& $k$-th eigenvector of $\widehat L_X$ (similarly $\widehat v_Y^{(k)}$). \\

$\widehat Q_X$
& $\mathbb{R}^{R \times R}$
& Empirical projection matrix onto the subspace orthogonal to the span of the leading $K$ eigenvectors of $\widehat L_X$ (similarly $\widehat Q_Y$). \\

$\widehat{L}^{\sf filt}_X$
& $\mathbb{R}^{R \times R}$
& Empirical filtered spectral operator associated with $\widehat L_X$ (similarly $\widehat{L}^{\sf filt}_Y$). \\

$\widehat L_d$
& $\mathbb{R}^{R \times R}$
& Empirical differential operator defined as $\widehat{L}^{\sf filt}_Y - \widehat{L}^{\sf filt}_X$. \\

$\widehat s(r)$
& $\mathbb{R}$
& Empirical region-wise score (test statistic) for region $r$. \\

$\widehat s_{\mathrm{PPI}}(r)$
& $\mathbb{R}$
& Region-wise score obtained from the PPI-based method. \\

$\widehat s_{\mathrm{UC}}(r,r')$
& $\mathbb{R}$
& Edge-wise score obtained from the UC method. \\

\bottomrule
\end{tabular}
\caption{Summary of notation. Population-level counterparts are denoted by removing hats.}
\label{tab:notation}
\end{table}
\section{Relevant Background}

Our method for identifying networks that differentiate between two states is based on constructing a separate graph for each state from the corresponding fMRI scans.
In this section, we briefly review two topics relevant to our approach.
Section \ref{sec. Graph-based differential analysis} introduces key definitions required to understand the spectral analysis of the dual graphs.
In our framework, the weight between two nodes represents the estimated FC between the corresponding brain regions.
Unlike previous studies, we estimate FC using distance correlation \citep{Szkely2007distance}.
Section \ref{sec. Distance correlation} formally defines distance correlation and explains why it is particularly well-suited for quantifying FC between brain regions.

\subsection{Spectral Graph Analysis}
\label{sec. Graph-based differential analysis}

FC, as estimated from fMRI, can be naturally represented using a graph framework. In this representation, each node corresponds to a distinct brain region, and the edge weights between pairs of nodes encode their FC. 
Specifically, for each condition, we construct an empirical undirected weighted graph with $R$ nodes corresponding to distinct brain regions. The pairwise statistical association between regions, typically derived from measures such as correlation, is encoded in the empirical symmetric weight matrices $\widehat W_X,\widehat W_Y \in \R^{R \times R}$.

For each graph, we compute a \emph{degree matrix}, $\widehat C_X, \widehat C_Y  \in \R^{R \times R}$, defined as diagonal matrices containing the total connectivity (degree) of each node. Specifically,
\[
\widehat C_X(r,r) = \sum_{r'=1}^R \widehat W_X(r,r') \, , \qquad 
\widehat C_Y(r,r) = \sum_{r'=1}^R \widehat W_Y(r,r') \, .
\]
Thus, $\widehat C_X(r,r)$ represents the total connections of brain region $r$ under condition $X$. In the context of brain networks, $\widehat C_X(r,r)$ can be interpreted as the total functional connectivity of brain region $r$ with the rest of the brain under condition $X$.

We construct the empirical \emph{symmetric normalized Laplacian matrices}, $\widehat L_X$ and $\widehat L_Y$, by
\begin{equation}
    \label{eq: laplacians}
   \widehat L_X = I_R - \widehat C_X^{-1/2} \widehat W_X \widehat C_X^{-1/2}\, , \qquad \widehat L_Y = I_R - \widehat C_Y^{-1/2} \widehat W_Y \widehat C_Y^{-1/2} \, ,
\end{equation}
where $I_R \in \R^{R \times R}$ denotes the identity matrix. The normalized Laplacian encodes the structural properties of the graph in a manner that is particularly suitable for spectral analysis \citep{chung1997spectral}.
Intuitively, each diagonal matrix $\widehat C_X$ (or $\widehat C_Y$) reflects the total connectivity of a node, while $\widehat L_X$ (or $\widehat L_Y$) balances this connectivity against the weights of its edges. This emphasizes relative connectivity patterns rather than raw connection strength, thereby highlighting global and local structural features of the network. In the context of brain networks, the empirical normalized Laplacian captures how strongly each brain region is connected relative to its overall connectivity profile, allowing us to characterize condition-specific differences in brain organization.

The eigenvectors of the empirical graph Laplacian, denoted by $\widehat v_X^{(1)}, \dots, \widehat v_X^{(R)}$ and $\widehat v_Y^{(1)}, \dots, \widehat v_Y^{(R)}$, are ordered according to their corresponding eigenvalues in non-decreasing order. These eigenvectors capture dominant modes of variation within the network \citep{von2007tutorial}. In particular, the first $K\leq R$ eigenvectors (associated with the smallest eigenvalues) provide a low-dimensional embedding of the graph that preserves its global connectivity structure \citep{belkin2003laplacian, coifman2006diffusion}. Such spectral representations have been widely applied in clustering \citep{ng2001spectral}, manifold learning \citep{belkin2003laplacian, coifman2006diffusion}, and network comparison \citep{wilson2008graph}. In our context, the spectral embeddings provide a compact feature representation for each brain region, enabling informative comparisons of connectivity patterns across conditions.

\subsection{Distance Correlation}\label{sec. Distance correlation}
Distance correlation is a measure of statistical dependence between random vectors \citep{Szkely2007distance}. In our setting, we use distance correlation to estimate FC between two brain regions. Let the BOLD activity in regions $r$ and $r'$ across $n$ subjects be represented by matrices  $X_r \in \R^{n \times T}$ and $X_{r'} \in \R^{n \times T}$, respectively. 
To estimate the distance correlation (dCor), we first compute the pairwise Euclidean distance matrices $D_r$ and $D_{r'}$, both of dimension $n \times n$.
Denote by $\mathds 1_n$ a vector of ones of size $n$, by $I_n$ an identity matrix of size $n \times n$, and by $H = I_n - \mathds 1_n \mathds 1_n^T/n$ the centering operator that projects onto the orthogonal complement of $\mathds 1_n$. The centered distance matrices are then given by $\tilde D_r = H D_r H$ and $\tilde D_{r'} = H D_{r'} H$, which ensures that both row and column sums are zero. 
For two matrices $A$ and $B$, let the inner product be defined as $\langle A,B\rangle = \sum_{i,j}A_{ij}B_{ij}$. The estimator of the distance correlation measure between $X_r$ and $X_{r'}$ is then
\begin{equation}\label{eq:dcor}
\widehat{\mathcal D}(X_r,X_{r'}) = \frac{\langle \tilde D_r,\tilde D_{r'} \rangle}{\sqrt{\langle \tilde D_r,\tilde D_r \rangle \langle \tilde D_{r'},\tilde D_{r'} \rangle}} \, .
\end{equation}
The distance correlation possesses several properties that make it particularly suitable for estimating FC between brain regions. 
At the population level, distance correlation is bounded between 0 and 1, and equals zero if and only if the two random vectors are statistically independent.
Since $\widehat{\mathcal D}(X_r, X_{r'})$ is based on Euclidean distances, it is invariant to rotations, scaling, and translations of $X_r$ and $X_{r'}$.
Most importantly, it also captures nonlinear dependencies between $X_r$ and $X_{r'}$, which linear correlation measures fail to detect. This flexibility is particularly relevant in fMRI settings, where relationships between regional BOLD signals may deviate from linear dependence. One limitation of distance correlation is its computational cost, which scales as $\OO(n^2)$ due to the computation of $D_r$ and $D_{r'}$.

\begin{remark}[Population and empirical distance correlation]
Assume that the rows of $X_r$ and $X_{r'}$ are i.i.d. paired samples
from a joint distribution on $\R^T \times \R^T$, with finite first
moments.
Under these conditions, the empirical distance correlation $\widehat{\mathcal D}(X_r,X_{r'})$ computed from $X_r$ and $X_{r'}$ converges to a population distance correlation $\mathcal D(r,r')$ as $n \to \infty$, see \cite[Theorem 2]{Szkely2007distance}. 

Throughout the paper, we use $\widehat W,\widehat L$ etc. to denote the empirical graph matrices computed from the data, and by $W,  L$ etc. their population counterparts.
In Section \ref{sec:theory}, 
we use the convergence of $\widehat{\mathcal D}(X_r,X_{r'})$ to 
$\mathcal D(r,r')$ to study the stability of the proposed spectral operators and region rankings.

\end{remark}

\section{SpARCD: Spectral Analysis for Revealing Connectivity Differences}\label{sec. sparcd}

We propose a method for identifying brain regions whose connectivity patterns differ between two states or conditions (e.g., faces versus shapes). 
The procedure consists of four steps (Figure~\ref{fig:flowchart}):
(I) construct a connectivity graph for each condition;
(II) apply spectral filtering to isolate differential structure; 
(III) compute region-level scores from the resulting differential operator; and 
(IV) assess significance using permutation testing and multiple-testing correction. We next describe each of these steps in detail. 




Step (I): For each pair of regions $(r, r')$, we compute the sample distance correlation between their respective signal matrices and use it as the edge weight connecting them. This yields two symmetric weight matrices,
\begin{equation}\label{eq:edge_weights}
\widehat W_X(r,r') = \widehat{\mathcal{D}}\big(X_r, X_{r'}\big) \, , \qquad 
\widehat W_Y(r,r') = \widehat{\mathcal{D}}\big(Y_r, Y_{r'}\big) \, ,
\end{equation}
corresponding to the two conditions under comparison, where $\widehat{\mathcal{D}}(\cdot,\cdot)$ denotes the sample distance correlation, as described earlier. These matrices estimate the pairwise association structure for each state. We then compute the corresponding empirical graph Laplacians $\widehat L_X$ and $\widehat L_Y$ using Eq.~\eqref{eq: laplacians}, which serve as our representations of functional connectivity in each state.


Step (II): Let $\widehat Q_X$ and $\widehat Q_Y$ denote the empirical projection matrices onto the subspaces orthogonal to the leading $K$ eigenvectors of $\widehat L_X$ and $\widehat L_Y$, respectively,
\begin{equation}\label{eq:projection}
    \widehat Q_X = I_R - \sum_{k=1}^K \widehat v_X^{(k)}\widehat v_X^{{(k)}\top } \, , \quad \widehat Q_Y = I_R - \sum_{k=1}^K \widehat v_Y^{(k)}\widehat v_Y^{{(k)} \top} \, .
\end{equation}

The leading eigenvectors of a graph Laplacian span the subspace corresponding to its most dominant (low-frequency) connectivity patterns~\citep{chung1997spectral}. Therefore, structural similarities between the graphs $\widehat L_X$ and $\widehat L_Y$ will tend to be captured within these top eigenvectors. To isolate differences, we construct empirical filtered operators by projecting each graph operator onto the orthogonal complement of the dominant eigenspace of the other graph
\[
\widehat{L}^{\sf filt}_Y = \widehat Q_X (I_R-\widehat L_Y) \widehat Q_X \, , \qquad \widehat{L}^{\sf filt}_X = \widehat Q_Y (I_R-\widehat L_X ) \widehat Q_Y \, .
\]
This projection removes components of $\widehat L_Y$ (or $\widehat L_X$) that lie in the direction of the leading connectivity patterns of $\widehat L_X$ (or $\widehat L_Y$), effectively filtering out shared structure and emphasizing differential connectivity.

Step (III): To quantify the difference in connectivity between two graphs, we compute the difference between their respective empirical differential operators via
\[
\widehat L_d = \widehat{L}^{\sf filt}_Y - \widehat{L}^{\sf filt}_X \, .
\]
The leading eigenvector of $\widehat L_d$, corresponding to the eigenvalue of $\widehat L_d$ with the largest absolute value and denoted by $\widehat v_d$, highlights the directions along which the two graphs differ the most, effectively capturing differences in connectivity structure between the two conditions. Since eigenvectors are defined only up to sign, we define the region-level test statistic by
\[
\widehat s(r) = \frac{|\widehat v_d(r)|}{\|\widehat v_d\|_1} \, ,
\]
where $\|\cdot\|_1$ denotes the $\ell_1$ norm. The absolute value removes the sign ambiguity of the eigenvector, while the $\ell_1$ normalization yields interpretable unit-sum scores across regions and emphasizes sparse differential structure. Regions with large values of $\widehat s(r)$ contribute strongly to the differential connectivity structure.

Step (IV): Under the null hypothesis, the connectivity structure is identical across the two conditions. To assess significance, we employ a nonparametric permutation procedure. For each permutation $\pi$, we generate permuted datasets $\{\tilde X,\tilde Y\}$ by reassigning condition labels and compute the corresponding test statistic $\widehat s^{(\pi)}(r)$. When the filtering parameter $K$ in Eq.~\eqref{eq:projection} is selected in a data-driven manner, it is re-estimated within each permutation to preserve the validity of the inference.

Under the null hypothesis and the corresponding exchangeability assumptions, the observed statistic $\widehat s(r)$ and its permutation counterpart $\widehat s^{(\pi)}(r)$ have the same distribution for every region $r$ and permutation $\pi$. Namely, 
\[
\widehat s(r) \overset d = \widehat s^{(\pi)}(r),
\qquad \mbox{for all } r \in \{1,\ldots,R\} \, , \quad \mbox{for all } \pi \in \{1,\ldots,B\} \, .
\]
The empirical $p$-value for region $r$ is computed as
\[
p_r =
\frac{1}{B+1}
\left\{
1+
\sum_{b=1}^B
\mathbb I
\left(
\widehat s^{(\pi)}(r)\ge \widehat s(r)
\right)
\right\}.
\]
To account for multiple testing across regions, we apply the Benjamini--Hochberg (BH) procedure at false discovery rate level $\alpha=0.05$.


{\bf Remark:} In the Hariri protocol, each participant completes alternating blocks of emotional (face) and neutral (shape) stimuli. Within each block, the BOLD time series exhibits strong temporal dependence and reflects a sustained cognitive state induced by the stimulus. Consequently, individual time points are not exchangeable, and permutation at the time-point level would violate the exchangeability assumptions required for valid permutation inference.
Instead, we treat each block as the basic unit of analysis. Under the null hypothesis that emotional and neutral conditions induce identical connectivity structure, block labels are exchangeable within each participant.   This motivates a block-level permutation scheme, in which the observed blocks are randomly reassigned to the two conditions separately for each participant.
This approach preserves the temporal coherence within blocks while respecting the experimental design and ensuring valid nonparametric inference, as shown in Theorem 1.

\subsection{Data-Driven Estimation of The Hyperparameter \texorpdfstring{$K$}{K}}\label{app:Data_driven K}
The hyperparameter $K$ controls the number of leading eigenvectors included in the construction of the projection operators $Q_X$ and $Q_Y$ (Step~II, Section~\ref{sec. sparcd}). 
Selecting too few eigenvectors (small $K$) retains too much shared structure between the graphs, potentially masking region-specific differences. Conversely, selecting too many eigenvectors (large $K$) may remove dominant connectivity patterns, thus reducing meaningful signal.
To guide the choice of $K$, we adopt a data-driven criterion that reflects the magnitude of regional differences captured by the test statistic $\widehat s(r)$. 
For each candidate $K \in \{1, 2, \dots, R\}$, we compute the score vector $\widehat s_{(K)}$, and its $\ell_2$ norm:
$\eta(K) = \|\widehat s_{(K)}\|_2$.
Since $\|\widehat s_{(K)}\|_1 = 1$ by construction, $\eta(K)$ serves as a measure of imbalance in $\widehat s_{(K)}$: higher values correspond to pronounced contrast between regions. We then select
$K^\ast = \arg\max_K \eta(K)$.
For the permutation test (Step~IV, Section~\ref{sec. sparcd}), $K^\ast$ is re-computed separately for each permutation to maintain the validity of the procedure.
In Section \ref{sec: application to hariri} we present a sensitivity analysis for $K$.

\begin{figure}[tb]
    \centering
    \includegraphics[width = 1.1\linewidth]{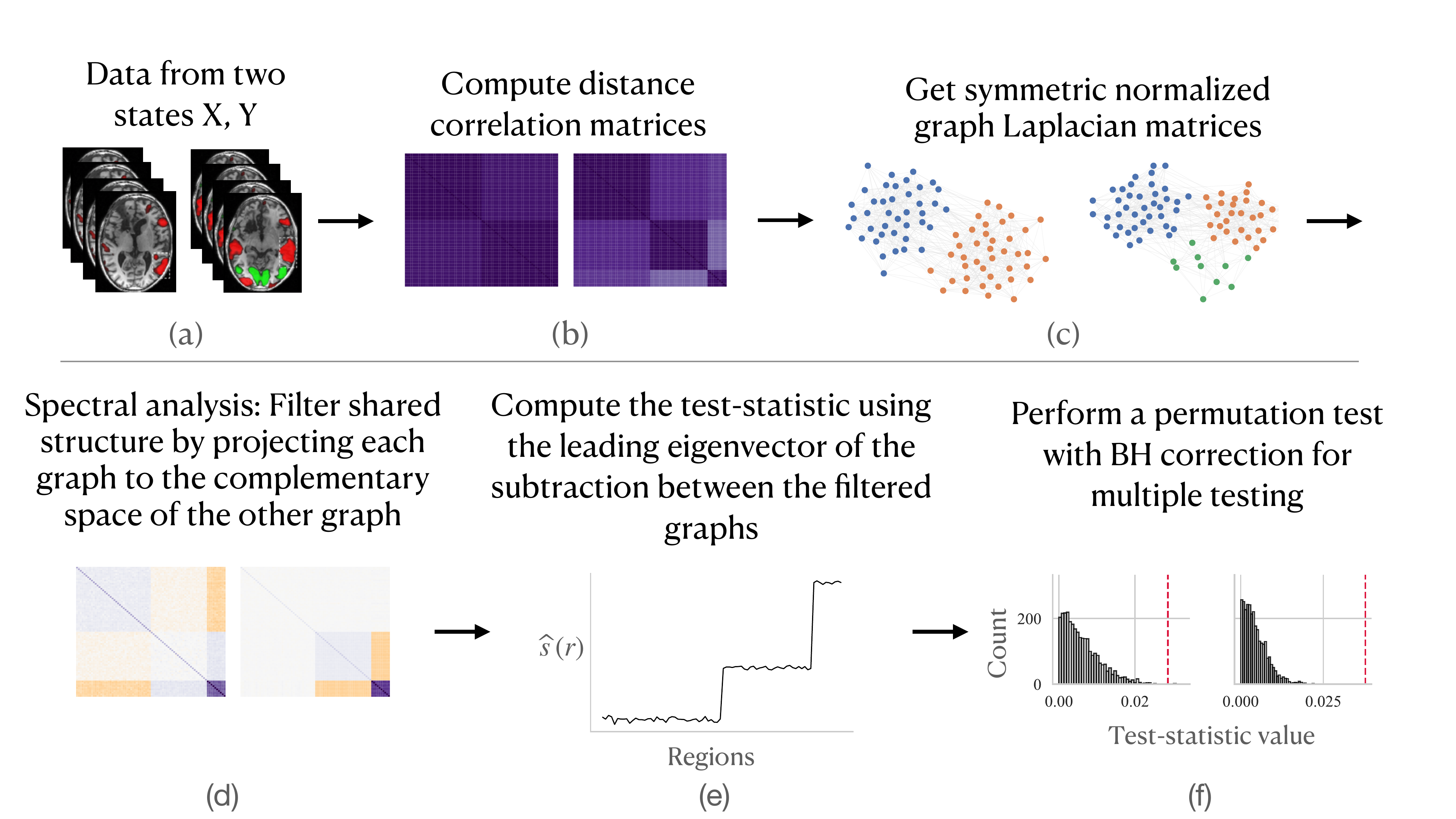}
    \caption{Flowchart describing the steps in the SpARCD algorithm.}
    \label{fig:flowchart}
\end{figure}

\subsection{Theoretical Guarantees}\label{sec:theory}

We provide two complementary theoretical guarantees for the proposed method. The first result establishes the finite-sample validity of the permutation test under a block-exchangeability assumption tailored to the experimental design. The second result provides an asymptotic stability analysis of the eigenvector-based region ranking, showing that the estimated scores and the resulting top-$K$ region set are robust to perturbations of the contrast operator. Together, these results justify both the inferential procedure (via valid $p$-values) and the stability of the resulting region ranking.

\begin{theorem}[Finite-sample validity of the block permutation test]
Fix a brain region \(r\), and let \(\widehat s(r)\) denote the observed test statistic computed from block data \(\mathcal B\) and condition labels \(A\). Let \(\mathcal G\) denote the set of allowable block permutations, and let \(\widehat s^{(\pi)}(r)\) denote the statistic computed from the permuted labels \(\pi A\), for \(\pi\in\mathcal G\). Assume that under the null hypothesis, the condition labels are exchangeable over \(\mathcal G\) conditional on the unlabeled block structure; that is,
\[
(\mathcal B,A)\mid \mathcal O
\overset d =
(\mathcal B,\pi A)\mid \mathcal O \, ,
\qquad \mbox{for all } \pi\in\mathcal G \, ,
\]
where
\[
\mathcal O=\{(\mathcal B,\pi A):\pi\in\mathcal G\}
\]
denotes the permutation orbit.
If the statistic is computed using the same deterministic procedure for both observed and permuted data, and if all permutations in \(\mathcal G\) are enumerated, then the exact permutation \(p\)-value
\[
p_r^*
=
\frac{1}{|\mathcal G|}
\sum_{\pi\in\mathcal G}
\mathbb I
\left\{
\widehat s^{(\pi)}(r)\ge \widehat s(r)
\right\}
\]
satisfies
$
{\Pr}_{H_0}(p_r^*\le\alpha\mid\mathcal O)\le\alpha
$
for every \(\alpha\in[0,1]\). Consequently, ${\Pr}_{H_0}(p_r^*\le\alpha)\le\alpha$.
\end{theorem}

The theorem considers the exact permutation $p$-value obtained by enumerating all permutations in $\mathcal G$.
In practice, enumerating all permutations in $\mathcal G$ is typically computationally infeasible. Therefore, the procedure described above approximates the exact permutation distribution using a large Monte Carlo sample of $B$ randomly generated permutations.

The following theorem establishes the stability of the proposed region-ranking procedure. It shows that if the estimated contrast matrix $\widehat L_d$ is close to its population counterpart $L_d$ in spectral norm and the leading eigenvalue of $L_d$ is sufficiently separated from the rest of the spectrum, then the leading eigenvector, the region-level scores, and the resulting top-$K$ ranking are stable under estimation error. Consequently, the eigengap plays an important practical role in determining ranking stability.

For a matrix $A$, let $\|A\|_2=\sup_{\|x\|_2=1}\|Ax\|_2$
denote the spectral norm. We assume the high-level consistency condition
$\|\widehat L_d-L_d\|_2=o_p(1)$,
which is justified by the consistency of empirical distance correlation \citep{Szkely2007distance,li2012feature} and the continuity of the graph and spectral transformations used to construct $L_d$.

\begin{theorem}[Stability of eigenvector-based region ranking]
\label{thm:eig-stability-ranking}

Let $R$ denote the number of brain regions, treated as fixed while $n\to\infty$. Let $L_d$
be the population contrast matrix, and let $\widehat L_d$ denote its empirical counterpart obtained from the proposed estimation procedure. Assume
\[
\Delta_n=\|\widehat L_d-L_d\|_2=o_p(1).
\]
Let $\lambda_1,\ldots,\lambda_R$ denote the eigenvalues of $L_d$, ordered so that
\[
|\lambda_1|>|\lambda_2|\ge\cdots\ge|\lambda_R| \, ,
\]
and assume that the leading absolute eigenvalue is separated from the rest of the spectrum. Specifically, define $\rho = |\lambda_1|-\max_{j \geq 2} |\lambda_j| >0$, and the eigengap is defined by $\varphi=\min_{j\ge2}|\lambda_1-\lambda_j|>0$.
Let $v_d$ be a unit-norm eigenvector associated with $\lambda_1$, and define
\[
s(r)=\frac{|v_d(r)|}{\|v_d\|_1} \, ,
\qquad r=1,\ldots,R \, .
\]
Then, as $n\to\infty$, the following statements hold with probability tending to one:
\begin{enumerate}

\item \textbf{Eigenvector stability.}
There exists a sign $\eta\in\{-1,1\}$ such that
\[
\|\widehat v_d-\eta v_d\|_2
\le
C_{\mathrm{DK}}
\frac{\Delta_n}{\varphi},
\]
where $C_{\mathrm{DK}}$ is a universal Davis--Kahan constant.

\item \textbf{Score stability.}
\[
\sup_{1\le r\le R}
|\widehat s(r)-s(r)|
=
O_p\!\left(
\frac{\Delta_n}{\varphi}
\right).
\]

\item \textbf{Ranking stability.}
If
\[
s(r_1)-s(r_2)
>
2\sup_{1\le r\le R}
|\widehat s(r)-s(r)|,
\]
then
\[
\widehat s(r_1)>\widehat s(r_2).
\]

\item \textbf{Top-$K$ stability.}
Let $S_K$ denote the set of the $K$ regions with largest population scores, and define
\[
\Delta_K=s_{(K)}-s_{(K+1)},
\]
where
\[
s_{(1)}\ge\cdots\ge s_{(R)}
\]
are the ordered population scores. If
\[
\Delta_K
>
2\sup_{1\le r\le R}
|\widehat s(r)-s(r)|,
\]
then the estimated top-$K$ set equals $S_K$.

\end{enumerate}
\end{theorem}
Theorem~\ref{thm:eig-stability-ranking} provides theoretical support for using the proposed scores and top-ranked regions as stable summaries of differential connectivity structure. The implications of Theorem~~\ref{thm:eig-stability-ranking} are illustrated empirically in Section~4 through the node-perturbation simulation, where Figure~\ref{fig: score separation} shows increasing separation between relevant and non-relevant nodes as the sample size grows, together with the spectrum of the empirical contrast operator. In the real-data analysis, Figure~\ref{fig: PTSD permutation} further reports the spectrum of $\widehat L_d$, providing a diagnostic assessment of the leading eigengap and the stability of the resulting region ranking.

\section{Simulation Study}\label{sec: simulation results}

In this section, we evaluate the ability of our approach to identify brain regions whose connectivity patterns differ between two states. We assess its performance using artificial datasets designed to mimic common challenges encountered in fMRI analysis. 

Across all experiments, we compare our method with three widely used approaches in the fMRI literature: 
(i) univariate correlation (UC) testing, which evaluates differences in pairwise correlations between conditions across all pairs of regions \citep{whitfield2012conn}. 
(ii) the Network-Based Statistic (NBS), a widely used approach for detecting subnetworks exhibiting significant differences between conditions. 
(iii) Psychophysiological Interaction (PPI), in which the FC is computed based on a predefined seed region \citep{friston1997psychophysiological}. 
In our simulation setting, PPI is applied under the simplifying assumption that states $A$ and $B$ correspond to two contiguous temporal blocks. While this block structure is not inherent to the data-generating mechanism, it is imposed solely to enable a controlled and comparable evaluation. To ensure a fair comparison, the PPI seed region is always selected from the ground-truth signal regions.
Additional implementation details for the UC, NBS, and PPI methods are provided in Appendix \ref{app:Methods description}.

\subsection{Data Generation}

We considered a sample size of $n=150$ and a signal length of $T=100$. 
We generated data under two complementary classes of connectivity structures: 
(i) block-structured models, and (ii) localized perturbation models. 
The former captures changes in community structure, while the latter introduces sparse or node-level modifications to the connectivity pattern.

We evaluate detection performance using four standard metrics: \textit{precision}, \textit{recall}, the \textit{F1 score}, and the \textit{area under the precision--recall curve (PR-AUC)}. These metrics quantify how accurately each method identifies truly differentiating brain regions, both when using continuous scores and at a fixed level of statistical significance. More details are available in Appendix \ref{app:simulations details}.

\paragraph*{Block-Structured Simulations}

We generated two sets of artificial signals, one for each state, with a block-dependency structure such that signals from two regions $r$ and $r'$ are statistically dependent if they belong to the same block and independent otherwise.
The block-dependency structures are identical for the two states, except for one block in state $X$, which is subdivided into two smaller blocks in state $Y$ (see top panel in Figure \ref{fig:simulation_blocks}). 
The objective is to correctly identify the regions comprising this modified block.
\begin{figure}[tb]
    \centering
    \includegraphics[width = 0.60\linewidth]{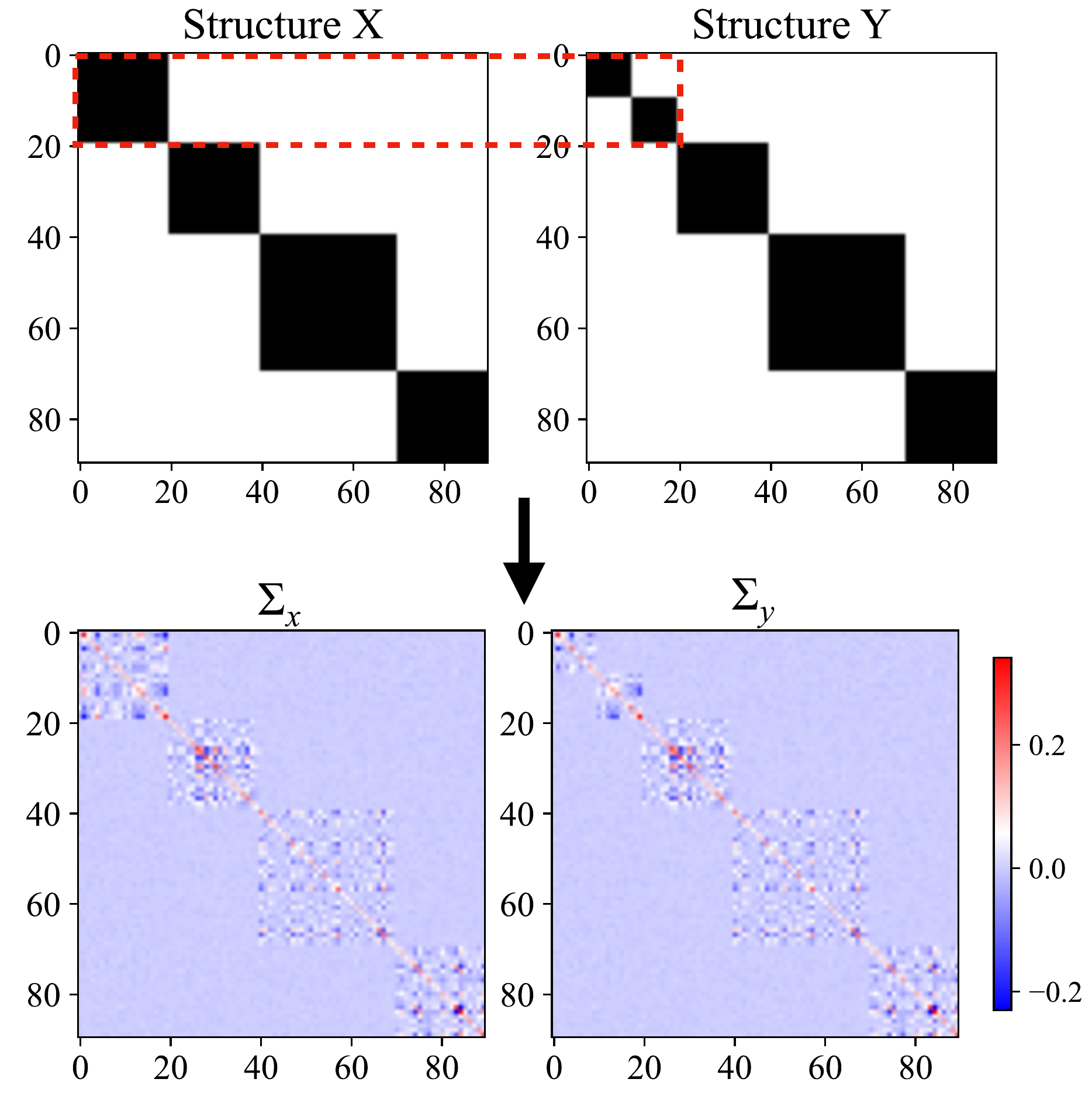}
    \caption{\textbf{Top panel:} An illustration for the block-dependency structure in states $X$ and $Y$. One block in $X$ is divided into two blocks in $Y$. \textbf{Bottom panel:} An example of the covariance matrices of the data used in the linear simulation.}
    \label{fig:simulation_blocks}
\end{figure}
We evaluated three types of statistical dependency profiles: 
\begin{enumerate}
    \item 

\textbf{Linear statistical dependency:} Each data column was independently sampled from a multivariate Gaussian distribution with zero mean and a block-diagonal covariance matrix $\Sigma_X \in \mathbb{R}^{R \times R}$. The block structure induces groups of correlated regions, while regions from different blocks remain independent. An illustration of $\Sigma_X$ and $\Sigma_Y$ is shown in the bottom panel of Fig. \ref{fig:simulation_blocks}.



\item \textbf{Nonlinear setting:} Each sample was generated from \(d\) independent latent seed signals \(S_X^{(i)}\sim\mathcal N(0,1)\). The observed regions were organized into blocks, and signals within each block were constructed as noisy nonlinear transformations of a shared seed signal, using sinusoidal functions with randomly generated frequencies and phases. This induces nonlinear dependence within blocks while maintaining independence between blocks. As in the linear setting, \(Y^{(i)}_{(\mathrm{nonlin})}\) was generated by modifying the block structure such that the last block of \(X\) was split into two blocks in \(Y\).



\item \textbf{Hybrid setting:} We generated signals that combine the linear and nonlinear dependency structure of the form
\[
X^{(i)} = \alpha  X^{(i)}_{(\text{lin})} + (1 - \alpha)  X^{(i)}_{(\text{nonlin})} + \epsilon \, ,
\]
where $\epsilon \sim \mathcal N(0,\sigma^2 I)$, and $\alpha \in [0,1]$ controls the relative contribution of the linear and nonlinear components. The matrices $X^{(i)}_{(\text{lin})}$ and $X^{(i)}_{(\text{nonlin})}$ are generated as described above. 
\end{enumerate}


\paragraph*{Localized Perturbation Simulations}
We consider two perturbation models that induce localized changes in the connectivity structure between states.

\textbf{Edge-wise perturbation:} In this setup, we modify a baseline covariance by increasing a selected subset of its entries by $\delta$. The baseline covariance is computed based on an rs-fMRI scan from the PTSD dataset. 
To ensure that the resulting covariance structure remains valid, we project it onto the positive semidefinite cone. Observations in dataset $X$ are sampled from a multivariate Gaussian distribution with zero mean and the baseline covariance, while dataset $Y$ is sampled using the perturbed covariance.
The ground truth in this case is defined based on the resulting correlation matrices.

\textbf{Node-wise perturbation:}
We examine a dynamic data-generating process based on a vector autoregressive (VAR) model, designed to simulate fMRI-like data in which brain regions demonstrate temporal dependence and network interactions.

In this framework, observations in dataset \(X\) are generated using a VAR(1) model in conjunction with a baseline covariance matrix, whereas dataset \(Y\) is produced with a perturbed covariance. We select a subset of nodes and adjust the baseline covariance matrix by strengthening connections within the subset while weakening those outside it.
The ground truth corresponds to the set of perturbed nodes.

\subsection{Results}

\paragraph*{Linear Setting:} 
\begin{figure}[tb]
    \centering
    \includegraphics[width = 0.9\linewidth]{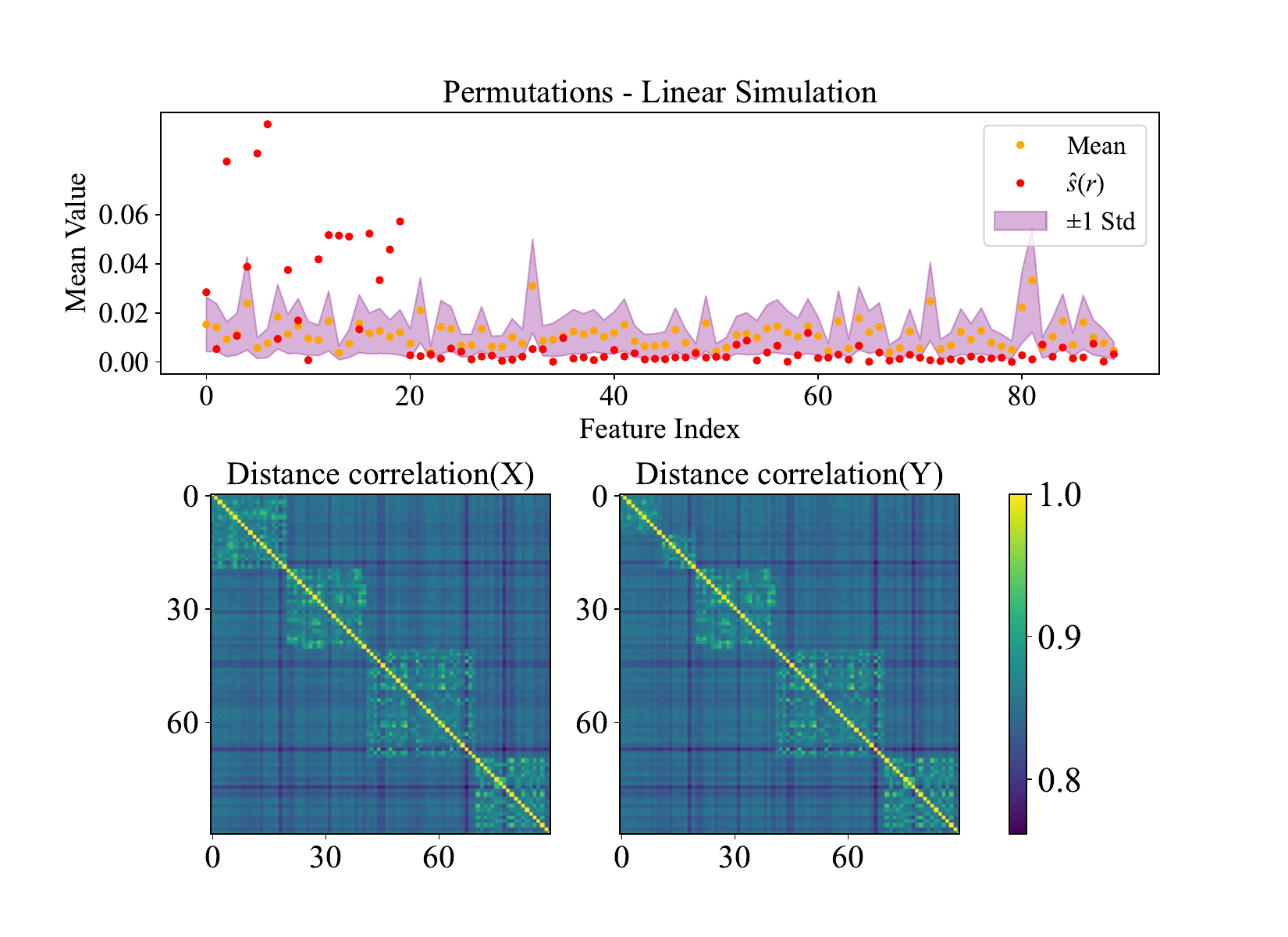}

    \caption{Linear simulation setting. \textbf{Top panel:} Observed test statistic $\widehat s(r)$ (red) compared with the mean (orange) and standard deviation (purple) of the permutation-based null distribution.
\textbf{Bottom panel:} Distance-correlation matrices for datasets $X$ and $Y$, highlighting the regions with altered connectivity in the first cluster.}
    \label{fig:linear simulation}
\end{figure}


    

Figure~\ref{fig:linear simulation} provides an illustration of our experiment's results. The upper panel shows the observed test statistic $\widehat s(r)$ in red, along with the mean (orange) and standard deviation (purple)  of the test statistics obtained from $B=1000$ random permutations. 
The bottom panel presents the distance-correlation matrices for $X$ and $Y$. As expected, the main connectivity differences are concentrated within the first block. The regions exhibiting the highest values of $\widehat s(r)$ closely align with this altered cluster, clearly emphasizing the relevant features  (indices 0--20). In contrast, the permuted statistics show no structural signal, confirming that the observed deviations are specific to the true connectivity differences.

To systematically evaluate performance, we repeated the simulation across a range of $\gamma$ values, where larger $\gamma$ induces stronger, more structured linear dependencies among the regions.
Figure~\ref{fig: f1 scores block setting} (left panel) shows the F1 score (after multiple-testing correction). We can see that NBS and UC yield very similar results and achieve the highest score across all $\gamma$ values. The results of PPI don't change with $\gamma$ and remain around 0.5. Our method shows a steeper improvement for $\gamma > 1$, but it still remains lower than NBS and UC due to low recall. Full results are provided in Appendix. \ref{app:simulations details}.

\paragraph*{Nonlinear Setting:} 

 

In the middle panel of Figure \ref{fig: f1 scores block setting}, it is evident that our method outperforms all other methods across all noise levels, denoted as $\sigma$. The observed decrease in performance is primarily due to a reduction in recall (for detailed results, refer to Appendix \ref{app:simulations details}).

\paragraph*{Hybrid Setting:} 


Figure \ref{fig: f1 scores block setting} (right panel) demonstrates that our method outperforms all baselines for low-to-moderate $\alpha$, where nonlinear effects dominate. As $\alpha$ increases and dependencies become predominantly linear, NBS and UC gradually improve and slightly exceed our results for $\alpha \geq 0.8$. PPI improves with $\alpha$, yet it remains consistently inferior, highlighting the limitations of seed-based inference in detecting distributed connectivity changes.

\begin{figure}[tb]
    \centering
    \includegraphics[width = 0.99\linewidth]{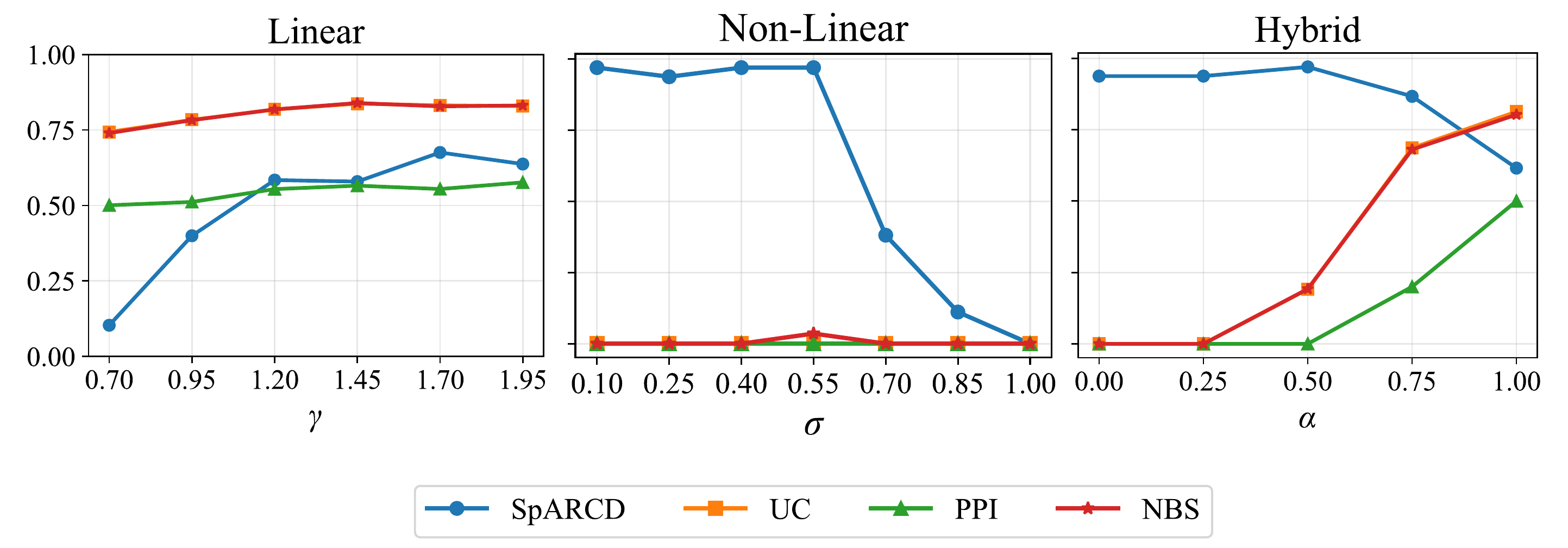}
    \caption{Performance of SpARCD in terms of F1 score and competing methods in the block-structure simulations. The results of the linear block setting against the parameter $\gamma$ (left panel), of the nonlinear block setting against the noise parameter $\sigma$ (middle panel), and of the hybrid block setting against the mixture parameter $\alpha$ (right panel).}
    \label{fig: f1 scores block setting}
\end{figure}


\paragraph{Edge-Wise Perturbation:}

Figure \ref{fig: f1 scores perturbation setting} shows that  NBS, followed by UC, performs best under this setting, achieving high scores as expected. Our method yields a medium score due to low recall, while PPI's performance is very low. More detailed results are presented in Appendix \ref{app:simulations details}.

\begin{figure}[tb]
    \centering
    \includegraphics[width = 0.8\linewidth]{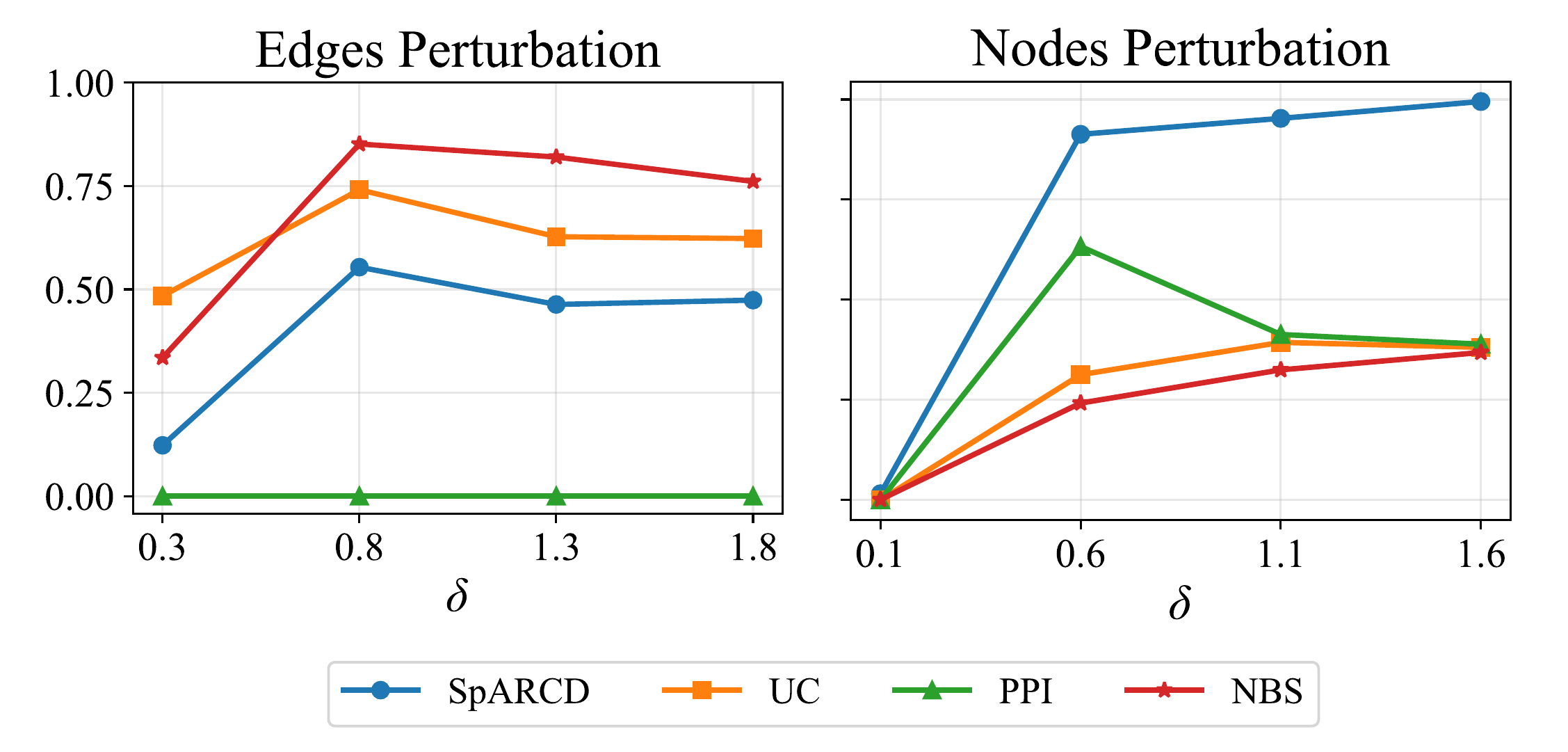}
    \caption{Performance of SpARCD in terms of F1 score and competing methods in the perturbation simulations. The results of the edges perturbation (left panel), and of the node-wise perturbations against the parameter $\delta$ (right panel).}
    \label{fig: f1 scores perturbation setting}
\end{figure}



\paragraph{Node-Wise Perturbation:}

Figure \ref{fig: f1 scores perturbation setting} demonstrates that our method performs best under this setting. UC, NBS, and PPI yield similar F1 scores for all $\delta$, much lower than our method.  It is worth noting that the recall and precision of UC and NBS vary from those of the PPI results. For a more detailed overview, refer to Appendix \ref{app:simulations details}.

Figure \ref{fig: score separation} illustrates the behavior of the proposed score as a function of the sample size $n$. For each $n$, we report the average score assigned to relevant (ground-truth) and non-relevant nodes across repeated simulations, along with standard deviation error bars. As $n$ increases, the score assigned to relevant nodes increases, while the score assigned to non-relevant nodes remains low, resulting in a clear and growing separation between the two groups. This result supports the consistency of the proposed method, demonstrating that it increasingly concentrates on the true relevant nodes as more data becomes available.


\begin{figure}[tb]
    \centering
    \includegraphics[width = 0.45\linewidth]{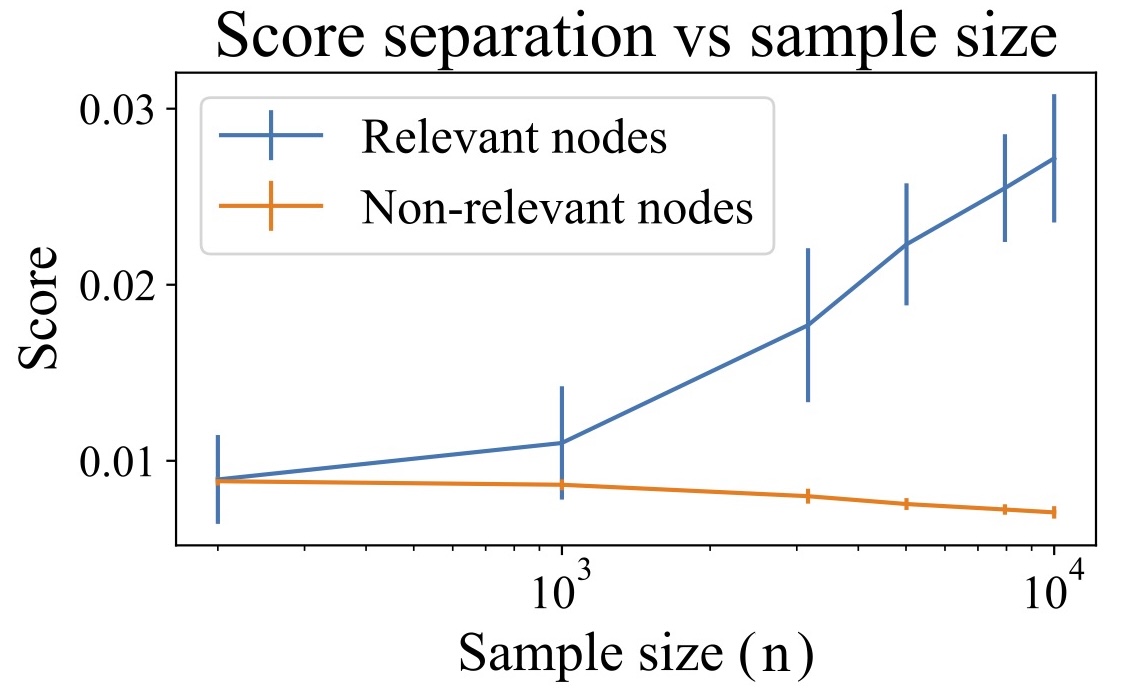}
    \includegraphics[width = 0.44\linewidth]{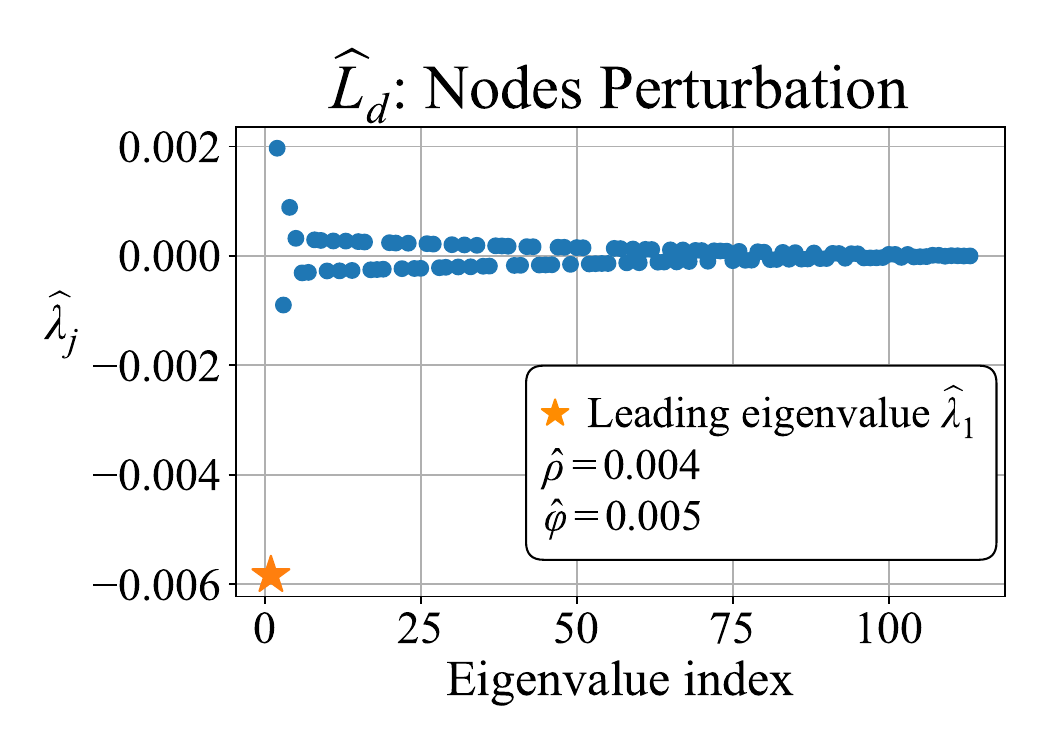}
    \caption{Empirical support for score stability in the node-perturbation setting. \textbf{Left panel:} Mean score assigned to relevant ground-truth nodes and non-relevant nodes as a function of the sample size $n$ on a logarithmic scale, averaged over 200 independent simulation runs. Error bars denote one standard deviation across repetitions. As $n$ increases, the proposed method assigns increasing mass to the relevant nodes while keeping the scores of non-relevant nodes low, yielding a clear separation between the two groups. These results provide empirical support for the consistency of the proposed estimator.
\textbf{Right panel:} Spectrum of $\widehat L_d$. The leading eigenvalue, $\lambda_1$, is highlighted by an orange star.}
    \label{fig: score separation}
\end{figure}

\section{EFMT and rs-fMRI Experiments}
We applied the SpARCD framework, together with competing methods, to two fMRI datasets from the cohort described in \cite{ben2019neurobehavioral}: task fMRI collected during the EFMT, following the Hariri protocol \citep{hariri2002amygdala}, and rs-fMRI. Specifically, we conducted four experiments:
\begin{enumerate}
    \item Comparing emotional vs. neutral conditions by analyzing task fMRI of PTSD subjects (Sec. \ref{sec: application to hariri}). 
    \item Comparing PTSD patients to the control group by analyzing rs-fMRI scans (Sec. \ref{sec: application to rfMRI}).
    \item Comparing emotional vs. neutral conditions by analyzing task fMRI of control group subjects (Appendix. \ref{subsec: addition efmt - non ptsd}).
    \item Comparing rs-fMRI to EFMT in PTSD subjects (App. \ref{subsec: addition efmt - rfmri vs efmt}).
\end{enumerate}
We summarize all the results of all four real-data experiments in Table \ref{tab:real_data_summary}.
 Table \ref{tab:harvard_roi_mapping} in the Appendix provides a mapping of regions to indices.

Across the four real-data experiments, SpARCD consistently identified regions associated with visual and sensory processing networks, suggesting stable condition-specific alterations in connectivity involving sensory systems. In particular, repeated detections were observed in occipital and posterior cortical regions, including the intracalcarine cortex (44–47), cuneal cortex (62–63), lingual gyrus (70–71), occipital fusiform regions (76–79), and posterior visual areas (92–95). 
Notably, substantial overlap was observed between the emotional vs. neutral comparison in experiments involving PTSD (experiment 1) and the control group (experiment 3). This indicates that similar sensory networks are implicated in the EFMT experiments in these two populations.  

\begin{center}
\begin{table*}[htb]
\centering

\renewcommand{\arraystretch}{1.2}
\begin{tabular} {p{5.5cm} p{6.5cm} p{3cm}} 
\hline
\textbf{Comparison} & \textbf{Detected Regions} & \textbf{\# Regions} \\
\hline

Emotional vs Neutral (PTSD)
& 44,45,62,63,76,77,78,79,94,95
& 10 \\

Emotional vs Neutral (control)
& 44,45,76, 77, 78, 79, 94, 95
& 8 \\

PTSD vs Control (r-state)
& 46,47,62,63,70,71,93,94
& 8 \\

EFMT vs r-state (PTSD)
& 62, 63, 77, 78, 79, 92, 93, 94, 95
& 9 \\

\hline
\end{tabular}
\caption{Summary of detected condition-specific regions across the four real-data experiments. 
Only regions surviving multiple-comparison correction are reported.}
\label{tab:real_data_summary}

\end{table*} 
\end{center}

\subsection{Paired Task-Based Application: Emotional versus Neutral Conditions in the Hariri EFMT Dataset} \label{sec: application to hariri}

We construct our dataset by segmenting each subject's BOLD time series into two groups of intervals. One group corresponds to the emotional-stimulus blocks, while the other group corresponds to the neutral-stimulus blocks. This results in two paired sets of fMRI recordings, $X, Y \in \R^{n \times R \times T}$. Here, $n=113$ represents the participants diagnosed with PTSD at baseline, $R=113$ indicates the number of ROIs, and $T=68$ denotes the time samples under each task type.


\begin{figure}[!htb]
\centering
\includegraphics[width=0.44\linewidth]{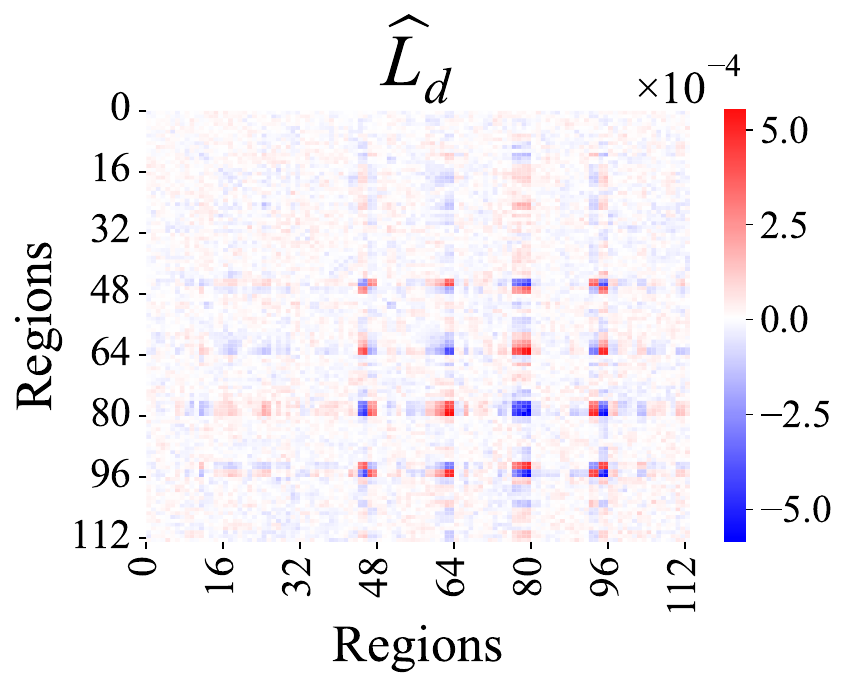}
\includegraphics[width=0.55\linewidth]{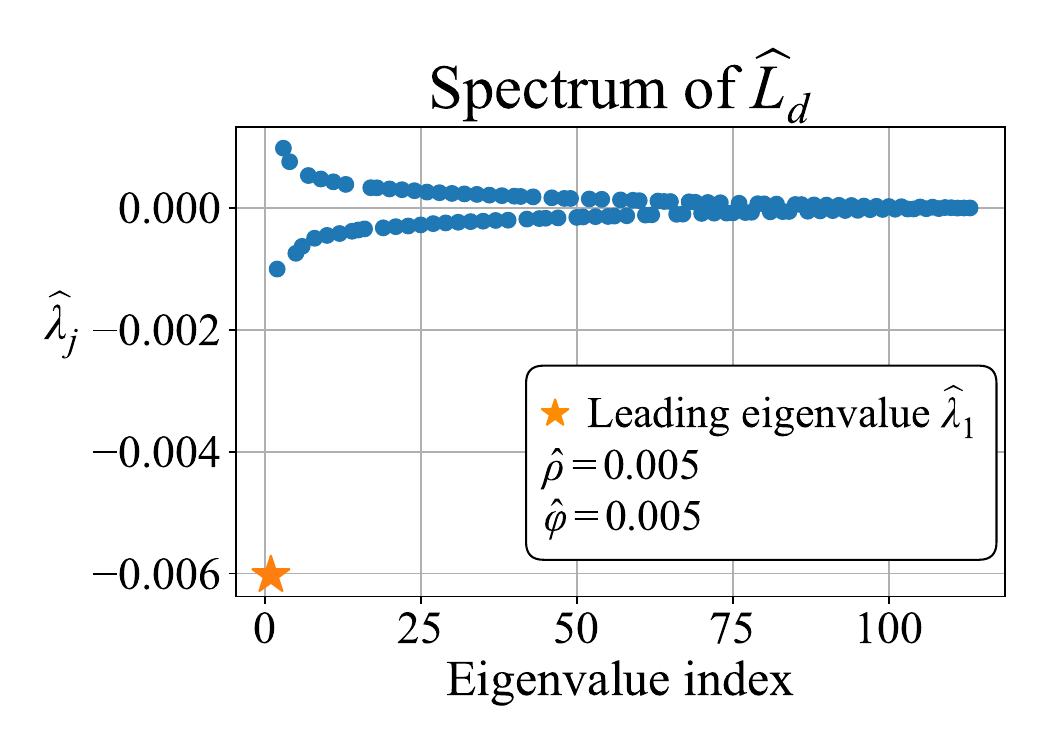}
\includegraphics[width=0.95\linewidth]{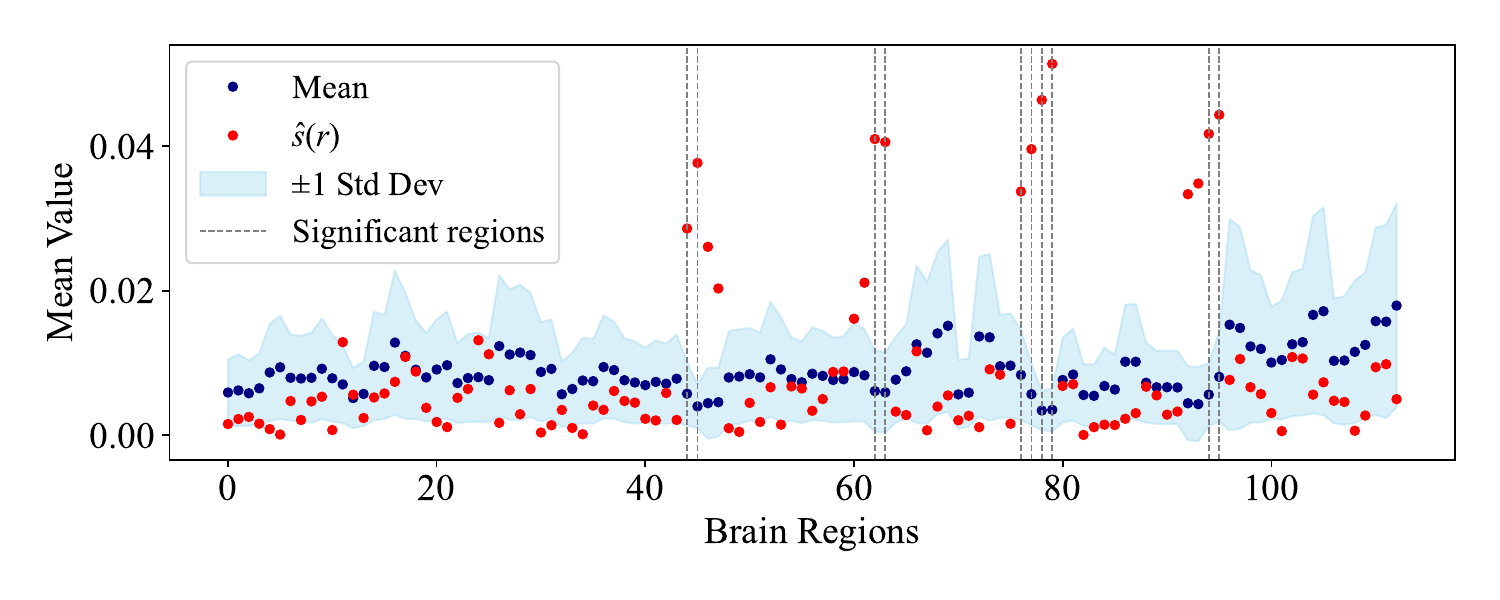}
\caption{Results of SpARCD applied to the Hariri fMRI dataset.
\textbf{Top left:} The empirical differential operator, $\widehat L_d$, on the Hariri fMRI dataset. {\bf Top right:} Spectrum of the empirical differential operator $\widehat L_d$. The orange star marks the leading eigenvalue, and the displayed eigengap supports the stability of the eigenvector-based region ranking.
\textbf{Bottom:} The observed test statistic $\widehat s(r)$ (red) is shown together with the mean (blue) and standard deviation (light blue) of the permutation-based null distribution. Several regions exhibit clear deviations above the null expectation, indicating significant differences in functional connectivity between the emotional and neutral task conditions.}
\label{fig: PTSD permutation}
\end{figure}

\begin{figure}[!htb]
\centering
\includegraphics[width=0.99\linewidth]{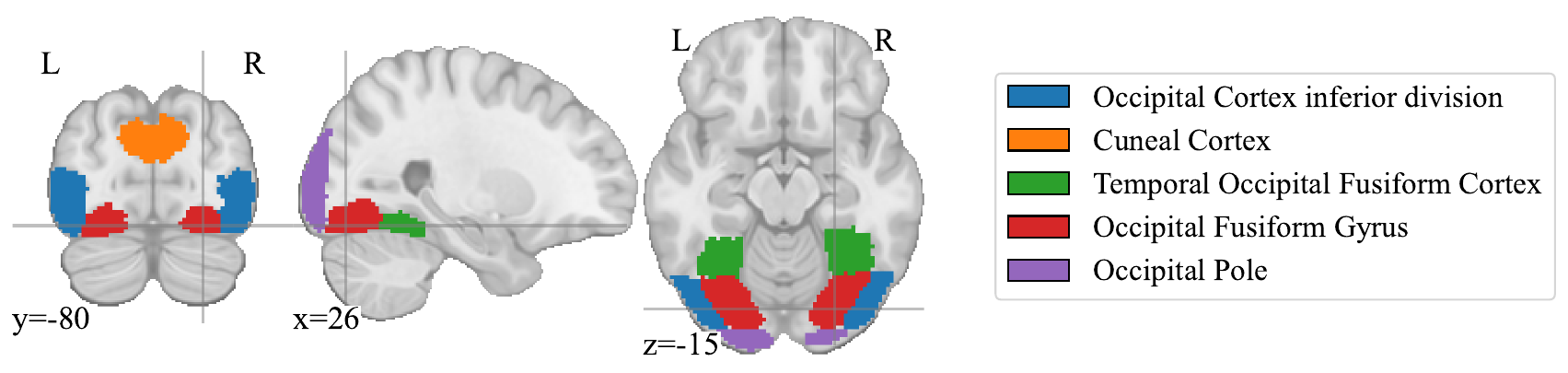}
\caption{Results of SpARCD applied to the Hariri fMRI dataset.
Anatomical locations of the identified significant regions are displayed on representative sagittal, coronal, and axial slices
at the indicated MNI coordinates. Different colors are used solely to distinguish separate
significant ROI. ”L” and ”R” denote the left and right hemispheres, respectively.}
\label{fig: hariri regions}
\end{figure}

After multiple-testing correction, SpARCD identified 10 significantly differentiating regions that primarily correspond to the bilateral lateral occipital cortex (44–45), cuneal cortex (62–63), temporal occipital fusiform cortex (76–77), occipital fusiform gyrus (78–79), and occipital pole (94–95) — areas within the posterior visual network, as can be seen in Figure~\ref{fig: hariri regions}. 
This constellation of posterior visual regions is consistent with prior studies implicating the visual association network in emotion perception and threat-related visual processing \citep{eklund2016cluster,barch2013function}.

Figure \ref{fig: PTSD permutation} (bottom) displays the empirical distribution of the test statistic $\widehat s(r)$ across regions. The red curve represents $\widehat s(r)$, while the blue curve and light-blue band are the mean and standard deviation obtained by $B=2000$ random permutations. Several regions exhibit marked deviations from the null expectation, indicating state-dependent alterations in FC. Figure~\ref{fig: PTSD permutation} (top right) displays the spectrum of the empirical contrast operator 
\(\widehat L_d\). The eigenvalue with the largest absolute value is negative, indicating that the dominant mode of differential connectivity is oriented in the negative direction of the contrast \(\widehat L_Y^{\mathrm{filt}}-\widehat L_X^{\mathrm{filt}}\). This sign has no 
effect on the region-level scores, which depend on the absolute entries of the 
corresponding eigenvector. More relevant for interpretation is the separation of this leading eigenvalue, in absolute value, from the remainder of the spectrum. This separation provides an empirical diagnostic, motivated by Theorem~\ref{thm:eig-stability-ranking}, that the 
leading eigenvector and the induced ranking of brain regions are not sensitive to near-ties among leading spectral components. The empirical eigengap $\widehat\varphi = 0.005$ is not small relative to the leading eigenvalue, suggesting that the leading spectral component is not the result of an apparent near-tie among several eigenvalues.


To characterize the specific brain networks with which these regions exhibit altered connectivity, one can analyze the operator $\widehat L_d$, which provides a descriptive summary of these changes, as demonstrated in Figure~\ref{fig: PTSD permutation} top left panel. In ongoing work, we plan to develop hierarchical FDR procedures to assign statistical significance to these network-level interpretations.

For comparison, we repeated the analysis using PPI, UC testing, and NBS. 
PPI results (Table \ref{tab: ppi} in Appendix \ref{app:Methods description}) depend strongly on the selected seed region: when the right amygdala (region 108) was used as the seed, 46 significant connections were identified; using the left amygdala (region 109) yielded only 34, with limited overlap beyond the limbic regions.
This variability underscores the seed-selection bias inherent in PPI and its reduced ability to detect distributed changes. 
The UC approach (Figure \ref{fig: UC NBS fMRI} left panel, and Table \ref{tab:sig_roi_connections} in Appendix \ref{app:Methods description}) produced a dense pattern of significant edges, 
mostly connected to regions overlapping with SpARCD’s findings.
NBS (Figure \ref{fig: UC NBS fMRI}, right panel in Appendix \ref{app:Methods description}) identified a single significant connected component ($p_{\mathrm{NBS}} < 1 \times 10^{-4}$), indicating the presence of a distributed connectivity alteration at the subnetwork level. However, the detected component was substantially broader and less spatially localized compared to the regions identified by SpARCD. 



\paragraph*{Sensitivity analysis for the parameter $K$}

To further assess the robustness of the data-driven selection rule in Sec. \ref{app:Data_driven K}, we conducted a sensitivity analysis for two experiments: (i) the emotional vs. neutral in EFMT, and (ii) Simulated linear-block data. In the EFMT experiments,  the data-driven choice was $K^\ast = 4$, and tended to yield a more comprehensive set of detected regions compared to fixed choices of $K$, as shown in Figure \ref{fig:sensitivity_analysis} (right panel). In the simulated data, for $K \geq 5$, the results were quite similar across different values of $K$, as shown in Figure \ref{fig:sensitivity_analysis} (left panel). For $K^\ast$, we got an average precision of 1, a recall of 0.43, and an F1 score of 0.59. 
Overall, the proposed criterion provides a balanced and stable choice of $K$, achieving competitive performance while avoiding manual tuning.



\begin{figure}[tb]
\centering

\begin{minipage}{0.55\textwidth}
    \centering
    \includegraphics[width=\linewidth]{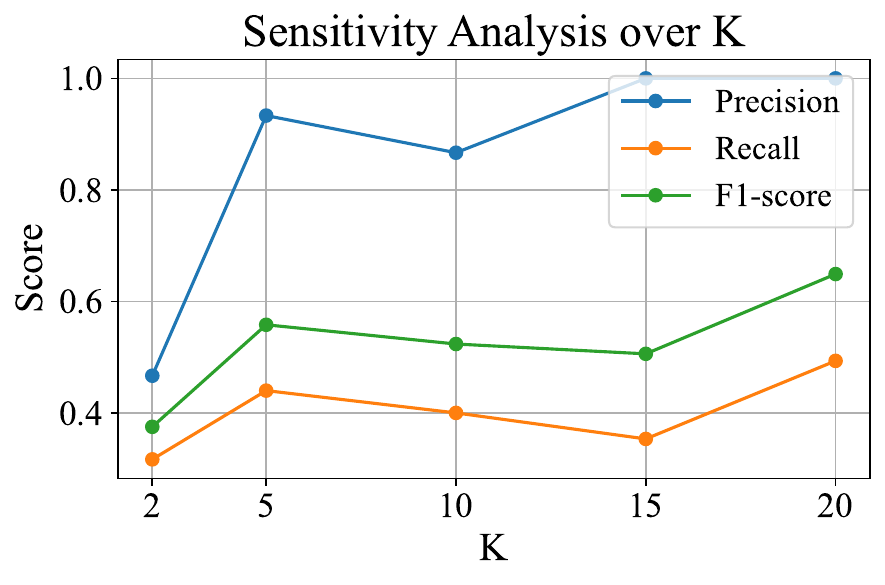}
\end{minipage}
\hfill
\begin{minipage}{0.4\textwidth}
    \centering
    \centering
\begin{tabular}{rlr}
\toprule
K & Detected Nodes & Count \\
\midrule
4 & 44, 45, 62, 63, 76, 77 & 10 \\
& 78, 79, 94, 95  \\
10 & 62, 63, 77, 78, 79 & 5 \\
15 & 50, 62, 63, 78, 79 & 5 \\
20 & 12, 37, 50, 62, 76, 78 & 6 \\
50 & 42, 62, 63 & 3 \\
\bottomrule
\end{tabular}
\end{minipage}

\caption{Sensitivity analysis for the parameter $K$ on the node-perturbation setting (left) and the EFMT experiment (right).}
\label{fig:sensitivity_analysis}

\end{figure}



\subsection{Unpaired Resting-State Application: PTSD versus Control Participants} \label{sec: application to rfMRI}

We applied SpARCD to resting-state fMRI data. As described in Section~\ref{sec: problem setting}, the dataset comprises $113$ participants diagnosed with PTSD and $42$ control subjects. For each participant, the scan consists of $T = 300$ time points, across $R = 113$ regions of interest (ROIs). 
In this experiment, our input consists of two \textit{unpaired} datasets $X \in \mathbb{R}^{113 \times 113 \times 300}$ and $Y \in \mathbb{R}^{42 \times 113 \times 300}$, corresponding to the PTSD and control groups, respectively.

As shown in Figure~\ref{fig: rfMRI PTSD vs Non}, multiple regions demonstrate clear deviations above the null expectation, suggesting altered connectivity patterns in the PTSD group. After multiple-testing correction, SpARCD identified eight regions with significant group-level differences, 
corresponding to the bilateral intracalcarine cortex (46–47), bilateral cuneal cortex (62–63), bilateral lingual gyrus (70-71), right supracalcarine cortex (93), and left occipital pole (94). The regions are illustrated in the bottom panel of Figure \ref{fig: rfMRI PTSD vs Non}. In this experiment, the eigengap was $\hat \varphi = 0.0079$. 

Similarly to the previous experiment, the red curve in Fig. ~\ref{fig: rfMRI PTSD vs Non} (top panel) represents the test statistic $\widehat s(r)$, while the blue curve and light-blue band are the mean and standard deviation obtained by $B=2000$ random permutations. 
In contrast to the paired EFMT analysis of Section 5, the resting-state comparison involves two independent groups. Accordingly, the permutation distribution was generated by randomly permuting subject-level group labels between the PTSD and control groups, preserving each participant's complete multivariate time series.

For comparison, we applied the UC and NBS baselines to the same dataset. After multiple-testing correction, the UC approach did not identify any significant edges. NBS detected four connected components; however, none survived component-level permutation correction at the $0.05$ significance level. PPI was not considered in this setting, since the present experiment involves unpaired resting-state data without an explicit task design or within-subject condition structure.

The fact that SpARCD detected significant effects while neither univariate testing nor NBS did suggests that PTSD-related connectivity differences are distributed across a coherent network of visual regions rather than concentrated in a few strongly altered edges. This finding is consistent with prior evidence implicating visual-processing areas in hypervigilance and threat monitoring in PTSD \citep{hendler2003sensing,korgaonkar2020intrinsic}. The spatially coherent pattern of the findings supports the ability of SpARCD to detect coordinated network-level reorganization.
These results should nevertheless be interpreted with caution due to the imbalance between the PTSD and control group sizes, which may affect the stability of the estimated connectivity structure. Additional validation on more balanced cohorts would therefore be valuable.



\begin{figure}[tb]

\centering
\includegraphics[width=0.9\linewidth]{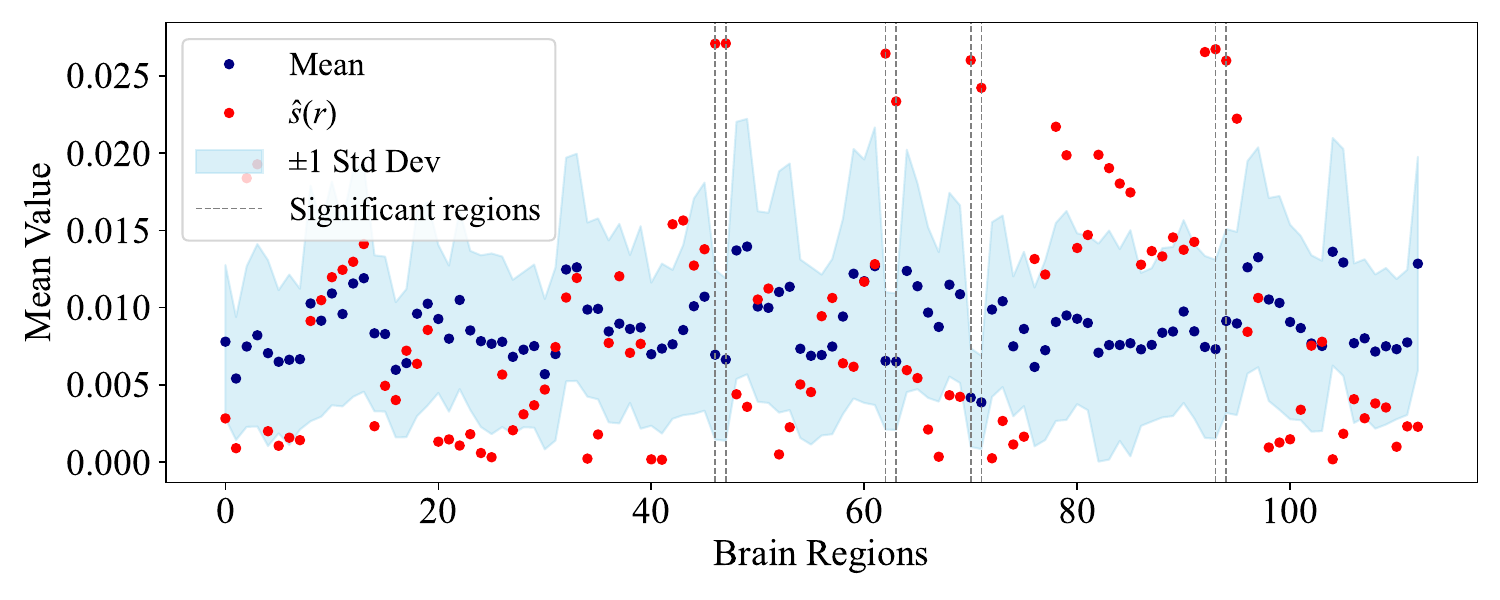}

\includegraphics[width=0.85\linewidth]{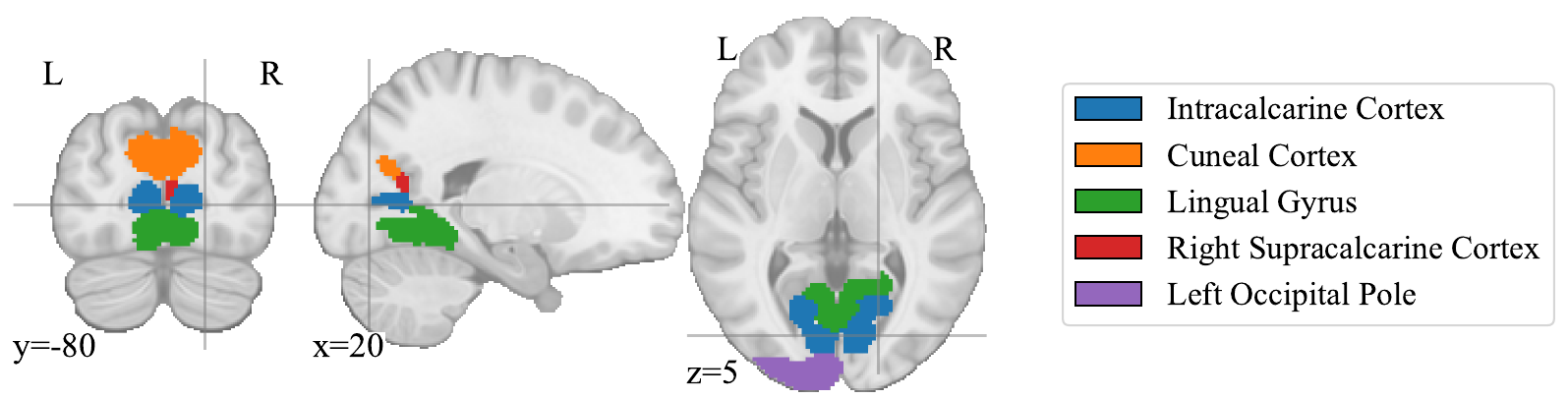}
\caption{Results of SpARCD applied to the resting-state fMRI dataset.
\textbf{Top panel:}
The observed test statistic $\widehat s(r)$ (red) is shown together with the mean (blue) and standard deviation (light blue) of the permutation-based null distribution. Several regions exhibit clear deviations above the null expectation, indicating significant differences in functional connectivity between PTSD subjects and controls.
\textbf{Bottom panel:} 
Anatomical locations of these significant regions are displayed on representative sagittal, coronal, and axial slices at the indicated MNI coordinates. Different colors are used solely to distinguish separate significant ROI. "L" and "R" denote the left and right hemispheres, respectively.}
\label{fig: rfMRI PTSD vs Non}
\end{figure}

\section{Discussion}

This work introduced SpARCD, a distance-correlation–based spectral framework for detecting differential functional connectivity between experimental conditions. By combining a nonlinear dependence measure with spectral graph filtering, SpARCD isolates structured, condition-specific changes in connectivity while suppressing shared baseline patterns. Through both simulations and empirical analysis of fMRI data, the method demonstrated high sensitivity and specificity in identifying distributed network alterations in paired and unpaired settings. The analysis was performed on EFMT (Hariri) and resting-state scans from PTSD and control patients, illustrating SpARCD's ability to address different types of fMRI scans. 


A central advantage of SpARCD lies in its ability to move beyond edge-wise testing toward network-level inference. Unlike traditional methods such as PPI or UC analysis, which test individual connections or seed-based interactions, SpARCD captures coherent, multi-regional connectivity shifts. The spectral projection step filters out components common to both conditions, thereby revealing only those eigenmodes that differ meaningfully between states. This mechanism yields a compact, interpretable subset of regions and mitigates false positives arising from widespread but nonspecific fluctuations.

In simulations, SpARCD consistently outperformed baseline methods across nonlinear and hybrid block-structured models and in node-wise perturbations. The method was particularly effective under nonlinear dependence, where classical correlation-based metrics lose power. These results highlight the value of distance correlation as a robust estimator of FC in complex neural systems.

When applied to the Hariri fMRI dataset, SpARCD uncovered a constellation of regions, predominantly within posterior visual areas, known to interact with limbic regions to support the processing of threat-related and emotional information. 
In the resting-state PTSD–control comparison, SpARCD identified posterior visual regions, including the intracalcarine and cuneal cortices, consistent with prior evidence of sensory-network dysregulation in PTSD. Even at rest, PTSD is characterized by atypical intrinsic coupling within early visual circuits, highlighting both the algorithm’s sensitivity and the functional relevance of these effects. Importantly, SpARCD detected these alterations through distributed connectivity changes rather than univariate activation differences, demonstrating its capacity to reveal subtle network-level dysregulation.

While the current study focused on binary condition comparisons, the framework naturally extends to multi-state or longitudinal designs, and can incorporate subject-level covariates through hierarchical modeling of spectral scores. Mathematically, the difference operator, or graph Laplacian, used to define $\widehat s(r)$ can be extended to model multiple covariance structures or repeated measurements over time. Computational scalability is ensured by efficient Laplacian decomposition and the use of parallelized distance-correlation estimation. Future work may integrate SpARCD with causal or directed connectivity models, enabling inference on information flow rather than undirected association.

In summary, SpARCD provides a statistically principled and computationally tractable tool for studying alterations in connectivity in high-dimensional neuroimaging data. Its general formulation makes it broadly applicable to other domains involving structured, multivariate dependence—such as genomics, climate networks, or financial systems, where uncovering subtle, distributed differences between complex systems is of primary interest.

\subsection*{Data Statements}

\paragraph{Ethics:}

This work is a secondary analysis of data previously collected and reported by \cite{ben2019neurobehavioral}. The original study was approved by the Ethics Committee of Tel Aviv Sourasky Medical Center (Reference No. 0207/14), and all participants provided written informed consent in accordance with the Declaration of Helsinki. The study was registered at ClinicalTrials.gov (Identifier: NCT03756545).

\paragraph{Data Availability:}

The data analyzed in this study will be available after QA and, once brought to maturity, by contacting the original study investigators on reasonable request. Code implementing the SpARCD method and reproducing the analyses is available at \url{https://github.com/ShiraAlon/SpARCD}.

\appendix

\section{Proofs for Theorems 1 and 2}
\begin{proof}[Proof of Theorem 1]
Condition on the permutation orbit \(\mathcal O\). Under the block-exchangeability
null, the observed labeling is uniformly distributed over the allowable relabelings
in \(\mathcal G\). Therefore, the observed statistic \(\widehat s(r)\) is equally likely to be
any element of the multiset $\{\widehat s^{(\pi)}(r):\pi\in\mathcal G\}$. Equivalently, the rank of the observed statistic among the permuted statistics is
uniform over the possible ranks, up to ties. The p-value
\[
p^*_r
=
|\mathcal G|^{-1}
\sum_{\pi\in\mathcal G}
\mathbf I\{\widehat s^{(\pi)}(r) \ge \widehat s(r)\}
\]
is the upper-tail randomization p-value associated with this conditional
permutation distribution. Hence, conditional on \(\mathcal O\),
\[
{\Pr}_{H_0}\{p^*_r\le \alpha\mid \mathcal O\}\le \alpha.
\]
Taking expectation over \(\mathcal O\) gives
\[
{\Pr}_{H_0}\{p^*_r\le \alpha\}\le \alpha.
\]
The same argument applies to the Monte Carlo version with the usual
\((1+\#)/(B+1)\) correction by exchangeability of the observed statistic and the
\(B\) sampled permutation statistics.
\end{proof}


\begin{proof}[Proof of Theorem 2]
We work on the event
$\mathcal E_n
=
\left\{
\Delta_n\le \min\left(\frac{\rho}{4},\frac{\varphi}{4}\right)
\right\}$.
Since \(\Delta_n=o_p(1)\) and \(\rho,\varphi>0\), we have
$\Pr(\mathcal E_n)\to1$.
By Weyl's inequality, every empirical eigenvalue \(\mu\in\sigma(\widehat L_d)\) is within \(\Delta_n\) of at least one population eigenvalue,
\[
\max_{\mu\in\sigma(\widehat L_d)}
\operatorname{dist}\{\mu,\sigma(L_d)\}
\le
\|\widehat L_d-L_d\|_2
=
\Delta_n .
\]
Since \(\lambda_1\) is separated from the rest of the spectrum by
$\varphi$,
and \(\Delta_n\le \varphi/4\), there is exactly one empirical eigenvalue in the
neighborhood
\[
(\lambda_1-\varphi/2 \, , \, \lambda_1+\varphi/2).
\]
Denote this eigenvalue by \(\widehat\lambda^\star\). Then
\[
|\widehat\lambda^\star-\lambda_1|\le \Delta_n.
\]
Every other empirical eigenvalue \(\mu\neq \widehat\lambda^\star\) cannot be
within \(\Delta_n\) of \(\lambda_1\), because otherwise it would also belong to
the isolated neighborhood of \(\lambda_1\). Hence, for every such \(\mu\), there
exists \(j\ge2\) such that
\[
|\mu-\lambda_j|\le \Delta_n.
\]
Consequently,
\[
|\mu|
\le
|\lambda_j|+\Delta_n
\le
\max_{j\ge2}|\lambda_j|+\Delta_n.
\]
On the event \(\Delta_n\le \rho/4\), 
we have
\[
|\widehat\lambda^\star|
\ge
|\lambda_1|-\Delta_n
>
\max_{j\ge2}|\lambda_j|+\Delta_n
\ge
|\mu|
\]
for every other empirical eigenvalue \(\mu\). Therefore the empirical eigenvalue
with largest absolute value is \(\widehat\lambda^\star\), and hence
\[
|\widehat\lambda_1-\lambda_1|\le \Delta_n.
\]
Next, let
$
P=v_dv_d^\top
$
be the rank-one orthogonal projector onto \(\operatorname{span}(v_d)\), and let
$
\widehat P=\widehat v_d\widehat v_d^\top
$
be the rank-one projector onto \(\operatorname{span}(\widehat v_d)\). Since
\(\lambda_1\) is separated from the rest of the spectrum by
$\varphi$ and since \(\Delta_n\le \varphi/4\) on \(\mathcal E_n\), the target eigenspace remains
isolated under the perturbation. The Davis--Kahan theorem gives
\[
\|\widehat P-P\|_2
\le
2\frac{\|\widehat L_d-L_d\|_2}{\varphi}
=
2\frac{\Delta_n}{\varphi}.
\]
For one-dimensional eigenspaces, there exists a sign \(\eta\in\{-1,1\}\) such that
\[
\|\widehat v_d-\eta v_d\|_2
\le
2\|\widehat P-P\|_2.
\]
Therefore,
\[
\|\widehat v_d-\eta v_d\|_2
\le
4\frac{\Delta_n}{\varphi}.
\]
Thus the eigenvector stability bound holds with \(C_{\mathrm{DK}}=4\).

The coordinate-wise bound follows immediately. For each region \(r\),
\[
\bigl||\widehat v_d(r)|-|v_d(r)|\bigr|
\le
|\widehat v_d(r)-\eta v_d(r)|
\le
\|\widehat v_d-\eta v_d\|_\infty
\le
\|\widehat v_d-\eta v_d\|_2.
\]
Hence
\[
\sup_{1\le r\le R}
\bigl||\widehat v_d(r)|-|v_d(r)|\bigr|
\le
C_{\mathrm{DK}}\frac{\Delta_n}{\varphi}.
\]
Now define
\[
a=\|v_d\|_1,
\qquad
\widehat a=\|\widehat v_d\|_1.
\]
By the reverse triangle inequality, and since the \(\ell_1\)-norm is \(\sqrt R\)-Lipschitz with respect to the
\(\ell_2\)-norm,
\[
|\widehat a-a|
=
\bigl|\|\widehat v_d\|_1-\|v_d\|_1\bigr|
\le
\|\widehat v_d-\eta v_d\|_1
\le
\sqrt R\,\|\widehat v_d-\eta v_d\|_2.
\]
Using the eigenvector bound,
\[
|\widehat a-a|
\le
C_{\mathrm{DK}}\sqrt R\,\frac{\Delta_n}{\varphi}.
\]
For each \(r\),
\[
|\widehat s(r)-s(r)|
=
\left|
\frac{|\widehat v_d(r)|}{\widehat a}
-
\frac{|v_d(r)|}{a}
\right|.
\]
Adding and subtracting \(|v_d(r)|/\widehat a\), we obtain
\[
|\widehat s(r)-s(r)|
\le
\frac{\bigl||\widehat v_d(r)|-|v_d(r)|\bigr|}{\widehat a}
+
|v_d(r)|
\left|
\frac{1}{\widehat a}-\frac{1}{a}
\right|.
\]
Equivalently, using a slightly more convenient denominator for the first term,
\[
|\widehat s(r)-s(r)|
\le
\frac{\bigl||\widehat v_d(r)|-|v_d(r)|\bigr|}{a}
+
|\widehat v_d(r)|
\frac{|\widehat a-a|}{a\widehat a}.
\]
Since \(|\widehat v_d(r)|\le \|\widehat v_d\|_2=1\), and if
\(\widehat a\ge a/2\), then
\[
|\widehat s(r)-s(r)|
\le
\frac{C_{\mathrm{DK}}\Delta_n}{\varphi a}
+
\frac{2C_{\mathrm{DK}}\sqrt R\,\Delta_n}{\varphi a^2}.
\]
Taking the supremum over \(r\) gives
\[
\sup_{1\le r\le R}|\widehat s(r)-s(r)|
\le
C_{\mathrm{DK}}\frac{\Delta_n}{\varphi}
\left(
\frac{1}{a}+\frac{2\sqrt R}{a^2}
\right).
\]
Because \(v_d\) is unit-norm, \(a=\|v_d\|_1\ge \|v_d\|_2=1\). Since
\(|\widehat a-a|\to0\) on \(\mathcal E_n\), the condition
\(\widehat a\ge a/2\) holds for all sufficiently large \(n\) on this event.

We now prove the ranking statements. Suppose that for two regions \(r_1,r_2\),
\[
s(r_1)-s(r_2)
>
2\sup_{1\le r\le R}|\widehat s(r)-s(r)|.
\]
Then
\[
\widehat s(r_1)-\widehat s(r_2)
=
\{s(r_1)-s(r_2)\}
+
\{\widehat s(r_1)-s(r_1)\}
-
\{\widehat s(r_2)-s(r_2)\}.
\]
Therefore,
\[
\widehat s(r_1)-\widehat s(r_2)
\ge
s(r_1)-s(r_2)
-
|\widehat s(r_1)-s(r_1)|
-
|\widehat s(r_2)-s(r_2)|
>0.
\]
Hence
\[
\widehat s(r_1)>\widehat s(r_2).
\]
Finally, let \(S_K\) be the population top-\(K\) set and define
\[
\Delta_K
=
\min_{r\in S_K,\ r'\notin S_K}
\{s(r)-s(r')\}
=
s_{(K)}-s_{(K+1)}.
\]
If
\[
\Delta_K>2\sup_{1\le r\le R}|\widehat s(r)-s(r)|,
\]
then every region in \(S_K\) has estimated score larger than every region outside
\(S_K\). Therefore the estimated top-\(K\) set equals \(S_K\).

All statements have been proved on \(\mathcal E_n\), and
\(\Pr(\mathcal E_n)\to1\). This completes the proof.
\end{proof}

\label{app1}

\section{Details on PPI NBS and UC.}\label{app:Methods description}

This section provides additional details on the Psycho-physiological Interaction (PPI) analysis, Univariate Correlation comparisons (UC), and the Network-Based Statistic (NBS), which are used as benchmarks for the proposed SpARCD method. 

\paragraph*{PPI}
The PPI approach \citep{friston1997psychophysiological} is a classical technique for identifying task-dependent changes in functional connectivity. PPI models the interaction between a psychological variable (task design) and a physiological variable (the BOLD signal from a predefined seed region) within a general linear model (GLM) framework.
Formally, for each brain region \( y(t) \), the model is:
\[
y(t) = \beta_0 + \beta_1 s(t) + \beta_2 p(t) + \beta_3 [s(t) \cdot p(t)] + \varepsilon(t),
\]
where \( s(t) \) denotes the seed region’s time series (physiological regressor), \( p(t) \) represents the psychological task regressor (e.g., emotional vs.\ neutral condition, convolved with the hemodynamic response function), and \( s(t) \cdot p(t) \) is their interaction term. The coefficient \( \beta_3 \) quantifies the task-dependent modulation of connectivity between the seed and each target region.
In our implementation, we computed the PPI regressor for each subject and seed region using the \texttt{nilearn} and \texttt{statsmodels} Python packages. The design matrix included the task, seed, and PPI terms, and a separate GLM was fitted for each target region. Group-level effects were assessed via one-sample \( t \)-tests across subjects, followed by FDR correction for multiple comparisons. Regions with FDR-corrected \( p < 0.05 \) were identified as exhibiting significant task-modulated connectivity with the seed. Table~\ref{tab: ppi} summarizes the brain regions that exhibited significant task-dependent connectivity with predefined seed regions under the PPI framework in the EFMT experiment in Section \ref{sec: application to hariri}.

\paragraph*{UC}
UC analysis identifies individual ROI-to-ROI connections that exhibit significant within-subject differences in functional connectivity across different conditions or states. In our application, for each subject and condition, we computed the Pearson correlation matrix of the BOLD time series across all regions, followed by Fisher's \( z \)-transformation to stabilize variance. The upper-triangular elements of each correlation matrix were then extracted, and paired-sample \( t \)-tests were conducted across subjects for each connection to assess condition-related differences.
Multiple comparison correction was performed using the Benjamini–Hochberg FDR procedure (\( \alpha = 0.05 \)). Connections with FDR-corrected \( p < 0.05 \) were considered significant and were mapped back to their corresponding ROI pairs to form a binary significance matrix. 
Table~\ref{tab:sig_roi_connections} and Figure \ref{fig: UC NBS fMRI} (left) present the results of the UC analysis in the EFMT experiment in Section \ref{sec: application to hariri}.

\begin{figure}[tb]
\centering
\includegraphics[width=0.99\linewidth]{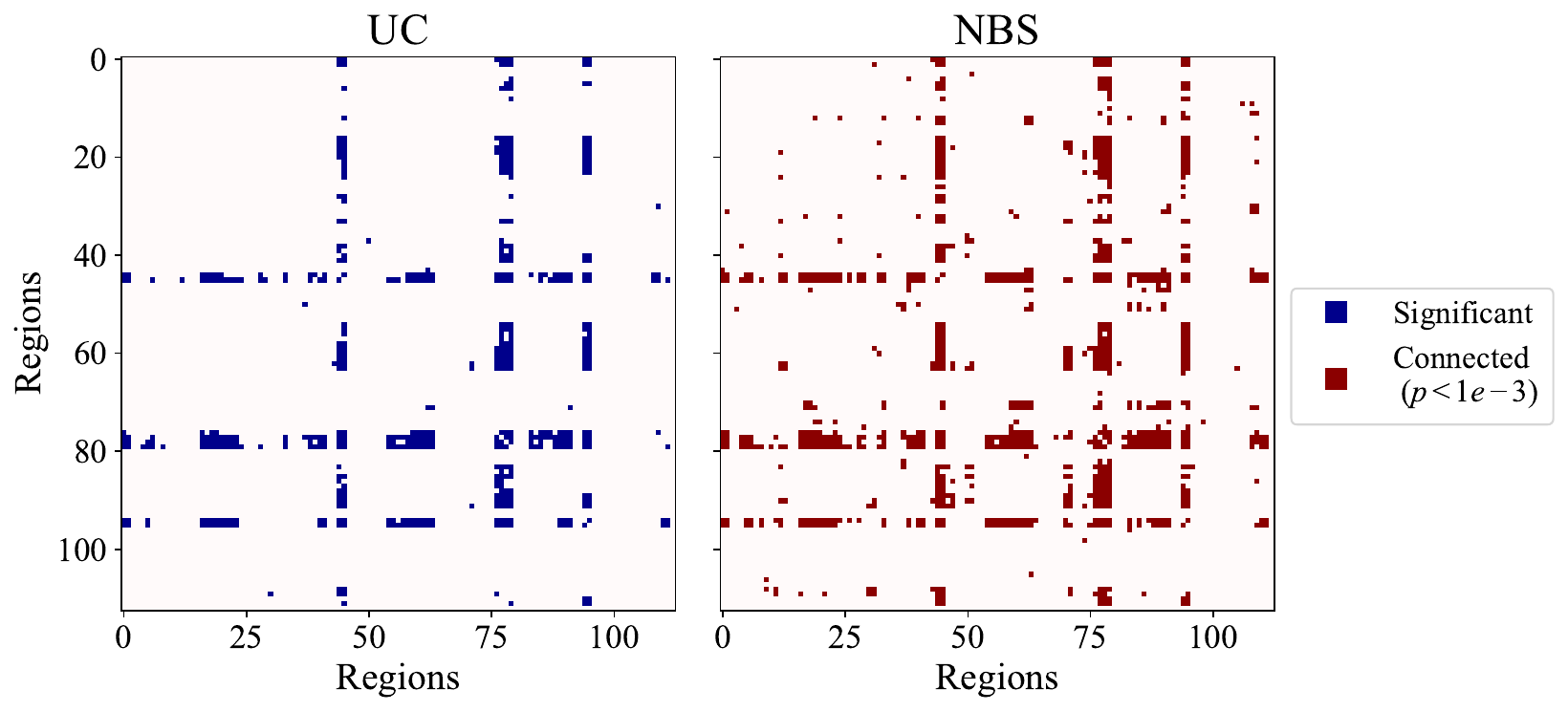}

\caption{Results of the UC and NBS methods applied to the Hariri fMRI dataset. \textbf{Left:} significant pairwise functional connections identified by the UC approach after BH multiple-testing correction, exhibiting a coherent spatial pattern of effects across the network. \textbf{Right:} NBS results, identifying a single significant connected component that is consistent with UC but includes more edges, forming a more explicitly connected network structure.}
\label{fig: UC NBS fMRI}
\end{figure}

\paragraph*{NBS}
The NBS \citep{zalesky2010network} is a graph-based method for identifying connected subnetworks that exhibit significant differences in connectivity while controlling the family-wise error rate (FWER). For each ROI-to-ROI connection, a paired-sample $t$-test was performed across subjects to compare connectivity values between task conditions. Connections exceeding a predefined primary threshold were retained, and connected components (subnetworks) were identified within the resulting suprathreshold graph. 
Statistical significance was assessed using permutation testing, in which condition labels were randomly permuted across subjects to generate a null distribution of the maximal connected-component size. Component-level p-values were obtained from this null distribution, and subnetworks were considered significant at $p<0.05$, controlling the FWER.
Significant subnetworks were then summarized by the set of ROIs participating in at least one significant connection.
We implemented NBS using the \texttt{nbs\_bct} function from the Brain Connectivity Toolbox Python package.
Figure \ref{fig: UC NBS fMRI} (right) presents the results of the NBS analysis in the EFMT experiment in Section \ref{sec: application to hariri}.

\begin{center}
\begin{table}[htb!]

\centering
\caption{Summary of PPI results by seed.}
\begin{tabularx}{\textwidth}{lX}
\toprule
\textbf{Seed 0} &
35 43 \hlroi{44} \hlroi{45} 77 79 \hlroi{94} \\
\midrule
\textbf{Seed 13} &
0 24 40 \hlroi{45} 51 \hlroi{76} \hlroi{77} \hlroi{78} \hlroi{79} 90 \hlroi{94}\\
\midrule
\textbf{Seed 77} &
0 1 2 3 4 5 6 7 8 9 10 11 12 13 14 15 16 17 18 19 20 21 22 23 24 25 26 27 28 29 30 31 32 33 34 35 36 37 38 39 40 41 42 43 45 46 47 48 50 51 52 53 54 55 56 57 58 59 60 61 62 63 64 65 66 67 68 69 70 71 72 73 74 75 78 79 80 81 82 83 84 85 86 87 88 89 90 91 92 93 94 95 96 97 98 99 100 101 102 103 104 105 106 107 108 110 111 112 \\
\midrule
\textbf{Seed 108} &
0 11 16 18 20 21 22 23 24 28 30 31 35 43 \hlroi{44} \hlroi{45} 54 55 56 57 58 59 60 70 71 \hlroi{76} \hlroi{77} \hlroi{78} \hlroi{79} \hlroi{94} \hlroi{95} 97 106 111 \\
\midrule
\textbf{Seed 109} &
0 1 2 11 15 16 17 18 19 20 21 22 23 26 27 28 30 31 43 \hlroi{44} \hlroi{45} 54 55 56 57 58 59 60 61 70 71 73 \hlroi{76} \hlroi{77} \hlroi{78} \hlroi{79} 84 85 86 87 \hlroi{94} \hlroi{95} 97 110 111 \\
\midrule
\textbf{Seed 112} & - 
\\
\bottomrule
\end{tabularx}

\label{tab: ppi}
\end{table}

\end{center}
\begin{center}
\begingroup
\footnotesize
\begin{longtable}{|c|p{0.8\textwidth}|}
\caption{Significant ROI-to-ROI connections after FDR correction.} \label{tab:sig_roi_connections} \\
\hline
ROI & Connected ROIs \\
\hline
\endfirsthead
\caption[]{Significant ROI-to-ROI connections after FDR correction.} \\
\hline
ROI & Connected ROIs \\
\hline
\endhead
\hline
\multicolumn{2}{r}{Continued on next page} \\
\hline
\endfoot
\hline
\endlastfoot
0 & \textcolor{red}{44}, \textcolor{red}{45}, \textcolor{red}{76}, \textcolor{red}{77}, \textcolor{red}{78}, \textcolor{red}{79}, \textcolor{red}{94}, \textcolor{red}{95} \\
1 & \textcolor{red}{44}, \textcolor{red}{45}, \textcolor{red}{77}, \textcolor{red}{78}, \textcolor{red}{79}, \textcolor{red}{94}, \textcolor{red}{95} \\
4 & \textcolor{red}{79} \\
5 & \textcolor{red}{78}, \textcolor{red}{79}, \textcolor{red}{94}, \textcolor{red}{95} \\
6 & \textcolor{red}{45}, \textcolor{red}{77}, \textcolor{red}{78}, \textcolor{red}{79} \\
8 & \textcolor{red}{79} \\
12 & \textcolor{red}{45} \\
16 & \textcolor{red}{44}, \textcolor{red}{45}, \textcolor{red}{76}, \textcolor{red}{77}, \textcolor{red}{78}, \textcolor{red}{79}, \textcolor{red}{94}, \textcolor{red}{95} \\
17 & \textcolor{red}{44}, \textcolor{red}{45}, \textcolor{red}{77}, \textcolor{red}{78}, \textcolor{red}{79}, \textcolor{red}{94}, \textcolor{red}{95} \\
18 & \textcolor{red}{44}, \textcolor{red}{45}, \textcolor{red}{76}, \textcolor{red}{77}, \textcolor{red}{78}, \textcolor{red}{79}, \textcolor{red}{94}, \textcolor{red}{95} \\
19 & \textcolor{red}{44}, \textcolor{red}{45}, \textcolor{red}{76}, \textcolor{red}{77}, \textcolor{red}{78}, \textcolor{red}{79}, \textcolor{red}{94}, \textcolor{red}{95} \\
20 & \textcolor{red}{44}, \textcolor{red}{45}, \textcolor{red}{77}, \textcolor{red}{78}, \textcolor{red}{79}, \textcolor{red}{94}, \textcolor{red}{95} \\
21 & \textcolor{red}{45}, \textcolor{red}{77}, \textcolor{red}{78}, \textcolor{red}{79}, \textcolor{red}{94}, \textcolor{red}{95} \\
22 & \textcolor{red}{45}, \textcolor{red}{77}, \textcolor{red}{78}, \textcolor{red}{79}, \textcolor{red}{94}, \textcolor{red}{95} \\
23 & \textcolor{red}{45}, \textcolor{red}{77}, \textcolor{red}{78}, \textcolor{red}{79}, \textcolor{red}{94}, \textcolor{red}{95} \\
24 & \textcolor{red}{45}, \textcolor{red}{79} \\
28 & \textcolor{red}{44}, \textcolor{red}{45}, \textcolor{red}{79} \\
29 & \textcolor{red}{45} \\
30 & 109 \\
33 & \textcolor{red}{44}, \textcolor{red}{45}, \textcolor{red}{77}, \textcolor{red}{78}, \textcolor{red}{79} \\
37 & 50, \textcolor{red}{77} \\
38 & \textcolor{red}{44}, \textcolor{red}{45}, \textcolor{red}{77}, \textcolor{red}{78}, \textcolor{red}{79} \\
39 & \textcolor{red}{44}, \textcolor{red}{77}, \textcolor{red}{79} \\
40 & \textcolor{red}{45}, \textcolor{red}{77}, \textcolor{red}{78}, \textcolor{red}{79}, \textcolor{red}{94}, \textcolor{red}{95} \\
41 & \textcolor{red}{44}, \textcolor{red}{45}, \textcolor{red}{77}, \textcolor{red}{78}, \textcolor{red}{79}, \textcolor{red}{94}, \textcolor{red}{95} \\
43 & \textcolor{red}{62} \\
\textcolor{red}{44} & 0, 1, 16, 17, 18, 19, 20, 28, 33, 38, 39, 41, \textcolor{red}{45}, 58, 59, 60, 61, \textcolor{red}{62}, \textcolor{red}{63}, \textcolor{red}{76}, \textcolor{red}{77}, \textcolor{red}{78}, \textcolor{red}{79}, 83, 85, 86, 88, 89, 90, 91, \textcolor{red}{94}, \textcolor{red}{95}, 108, 109 \\
\textcolor{red}{45} & 0, 1, 6, 12, 16, 17, 18, 19, 20, 21, 22, 23, 24, 28, 29, 33, 38, 40, 41, \textcolor{red}{44}, 54, 55, 56, 58, 59, 60, 61, \textcolor{red}{62}, \textcolor{red}{63}, \textcolor{red}{76}, \textcolor{red}{77}, \textcolor{red}{78}, \textcolor{red}{79}, 85, 87, 88, 89, 90, 91, \textcolor{red}{94}, \textcolor{red}{95}, 108, 109, 111 \\
50 & 37 \\
54 & \textcolor{red}{45}, \textcolor{red}{77}, \textcolor{red}{78}, \textcolor{red}{79}, \textcolor{red}{94}, \textcolor{red}{95} \\
55 & \textcolor{red}{45}, \textcolor{red}{77}, \textcolor{red}{78}, \textcolor{red}{79}, \textcolor{red}{94}, \textcolor{red}{95} \\
56 & \textcolor{red}{45}, \textcolor{red}{77}, \textcolor{red}{79}, \textcolor{red}{95} \\
57 & \textcolor{red}{77}, \textcolor{red}{79}, \textcolor{red}{94}, \textcolor{red}{95} \\
58 & \textcolor{red}{44}, \textcolor{red}{45}, \textcolor{red}{77}, \textcolor{red}{78}, \textcolor{red}{79}, \textcolor{red}{94}, \textcolor{red}{95} \\
59 & \textcolor{red}{44}, \textcolor{red}{45}, \textcolor{red}{76}, \textcolor{red}{77}, \textcolor{red}{78}, \textcolor{red}{79}, \textcolor{red}{94}, \textcolor{red}{95} \\
60 & \textcolor{red}{44}, \textcolor{red}{45}, \textcolor{red}{76}, \textcolor{red}{77}, \textcolor{red}{78}, \textcolor{red}{79}, \textcolor{red}{94}, \textcolor{red}{95} \\
61 & \textcolor{red}{44}, \textcolor{red}{45}, \textcolor{red}{76}, \textcolor{red}{77}, \textcolor{red}{78}, \textcolor{red}{79}, \textcolor{red}{94}, \textcolor{red}{95} \\
\textcolor{red}{62} & 43, \textcolor{red}{44}, \textcolor{red}{45}, 71, \textcolor{red}{76}, \textcolor{red}{77}, \textcolor{red}{78}, \textcolor{red}{79}, \textcolor{red}{94}, \textcolor{red}{95} \\
\textcolor{red}{63} & \textcolor{red}{44}, \textcolor{red}{45}, 71, \textcolor{red}{76}, \textcolor{red}{77}, \textcolor{red}{78}, \textcolor{red}{79}, \textcolor{red}{94}, \textcolor{red}{95} \\
71 & \textcolor{red}{62}, \textcolor{red}{63}, 91 \\
\textcolor{red}{76} & 0, 16, 18, 19, \textcolor{red}{44}, \textcolor{red}{45}, 59, 60, 61, \textcolor{red}{62}, \textcolor{red}{63}, \textcolor{red}{77}, \textcolor{red}{78}, \textcolor{red}{79}, 83, 85, 86, 87, 89, 90, 91, \textcolor{red}{94}, \textcolor{red}{95}, 109 \\
\textcolor{red}{77} & 0, 1, 6, 16, 17, 18, 19, 20, 21, 22, 23, 33, 37, 38, 39, 40, 41, \textcolor{red}{44}, \textcolor{red}{45}, 54, 55, 56, 57, 58, 59, 60, 61, \textcolor{red}{62}, \textcolor{red}{63}, \textcolor{red}{76}, \textcolor{red}{78}, \textcolor{red}{79}, 83, 84, 85, 86, 87, 88, 89, 90, 91, \textcolor{red}{94}, \textcolor{red}{95} \\
\textcolor{red}{78} & 0, 1, 5, 6, 16, 17, 18, 19, 20, 21, 22, 23, 33, 38, 40, 41, \textcolor{red}{44}, \textcolor{red}{45}, 54, 55, 58, 59, 60, 61, \textcolor{red}{62}, \textcolor{red}{63}, \textcolor{red}{76}, \textcolor{red}{77}, \textcolor{red}{79}, 83, 85, 88, 89, 90, 91, \textcolor{red}{94}, \textcolor{red}{95} \\
\textcolor{red}{79} & 0, 1, 4, 5, 6, 8, 16, 17, 18, 19, 20, 21, 22, 23, 24, 28, 33, 38, 39, 40, 41, \textcolor{red}{44}, \textcolor{red}{45}, 54, 55, 56, 57, 58, 59, 60, 61, \textcolor{red}{62}, \textcolor{red}{63}, \textcolor{red}{76}, \textcolor{red}{77}, \textcolor{red}{78}, 83, 84, 85, 88, 89, 90, 91, \textcolor{red}{94}, \textcolor{red}{95}, 111 \\
83 & \textcolor{red}{44}, \textcolor{red}{76}, \textcolor{red}{77}, \textcolor{red}{78}, \textcolor{red}{79} \\
84 & \textcolor{red}{77}, \textcolor{red}{79} \\
85 & \textcolor{red}{44}, \textcolor{red}{45}, \textcolor{red}{76}, \textcolor{red}{77}, \textcolor{red}{78}, \textcolor{red}{79} \\
86 & \textcolor{red}{44}, \textcolor{red}{76}, \textcolor{red}{77} \\
87 & \textcolor{red}{45}, \textcolor{red}{76}, \textcolor{red}{77} \\
88 & \textcolor{red}{44}, \textcolor{red}{45}, \textcolor{red}{77}, \textcolor{red}{78}, \textcolor{red}{79} \\
89 & \textcolor{red}{44}, \textcolor{red}{45}, \textcolor{red}{76}, \textcolor{red}{77}, \textcolor{red}{78}, \textcolor{red}{79}, \textcolor{red}{94}, \textcolor{red}{95} \\
90 & \textcolor{red}{44}, \textcolor{red}{45}, \textcolor{red}{76}, \textcolor{red}{77}, \textcolor{red}{78}, \textcolor{red}{79}, \textcolor{red}{94}, \textcolor{red}{95} \\
91 & \textcolor{red}{44}, \textcolor{red}{45}, 71, \textcolor{red}{76}, \textcolor{red}{77}, \textcolor{red}{78}, \textcolor{red}{79}, \textcolor{red}{94}, \textcolor{red}{95} \\
\textcolor{red}{94} & 0, 1, 5, 16, 17, 18, 19, 20, 21, 22, 23, 40, 41, \textcolor{red}{44}, \textcolor{red}{45}, 54, 55, 57, 58, 59, 60, 61, \textcolor{red}{62}, \textcolor{red}{63}, \textcolor{red}{76}, \textcolor{red}{77}, \textcolor{red}{78}, \textcolor{red}{79}, 89, 90, 91, \textcolor{red}{95}, 110, 111 \\
\textcolor{red}{95} & 0, 1, 5, 16, 17, 18, 19, 20, 21, 22, 23, 40, 41, \textcolor{red}{44}, \textcolor{red}{45}, 54, 55, 56, 57, 58, 59, 60, 61, \textcolor{red}{62}, \textcolor{red}{63}, \textcolor{red}{76}, \textcolor{red}{77}, \textcolor{red}{78}, \textcolor{red}{79}, 89, 90, 91, \textcolor{red}{94}, 110, 111 \\
108 & \textcolor{red}{44}, \textcolor{red}{45} \\
109 & 30, \textcolor{red}{44}, \textcolor{red}{45}, \textcolor{red}{76} \\
110 & \textcolor{red}{94}, \textcolor{red}{95} \\
111 & \textcolor{red}{45}, \textcolor{red}{79}, \textcolor{red}{94}, \textcolor{red}{95} \\
\end{longtable}
\endgroup
\end{center}

\section{Details About Simulations}\label{app:simulations details}

\paragraph{Linear Setting}
Let $l_X = \{l_X^{(1)}, \ldots, l_X^{(q)}\}$ denote the sizes of the $q$ dependency blocks in $X$. 

Let $\Sigma_X^{j} \in \mathbb{R}^{l_X^{(j)} \times l_X^{(j)}}$ be the submatrix of $\Sigma_X$ corresponding to the $j$-th block. Each block covariance is constructed as $\Sigma_X^j = U^{j} \Delta^{j} (U^{j})^\top$,
where $U^{j}$ is a randomly drawn orthogonal matrix and $\Delta^{j}$ is a diagonal matrix with entries 
$$
\Delta^j(k,k) = \left( \frac{1}{k} \right)^{\gamma} \, ,\qquad k = 1,\ldots,l^{(j)} \, .
$$
The parameter $\gamma$ controls the spectral decay within each block and thus determines the level of intra-block correlation. For example, setting $\gamma = 0$ yields the identity matrix, which corresponds to uncorrelated elements.

\textbf{Parameters for linear setting:} To increase the problem difficulty, we introduce heterogeneity across samples, such that for each sample, the block size vector $\ell_X$ was sampled from one out of the following three configurations,
\[
\{20, 20, 30, 20\}, \qquad \{20, 21, 29, 20\}, \qquad \{20, 19, 31, 20\}
\]
In $Y$, the first block was partitioned into two identical blocks. The partition was selected from one of the three following configurations
\[
\{10,10\},\{9,11\},\{11,9\}.
\]

\textbf{Full results:}
Figure~\ref{fig:linear simulation results} (upper left) depicts the PR-AUC as a function of $\gamma$. Although our method initially yields slightly lower PR-AUC at weak signal strengths, it improves rapidly and surpasses both competing approaches once  $\gamma \approx 0.9$. This indicates that as the linear dependency structure strengthens, our framework more effectively prioritizes the truly relevant regions. The other panels of Figure \ref{fig:linear simulation results} show the precision, recall, and F1 score (after multiple-testing correction). While NBS, UC, and PPI maintain relatively stable precision across all $\gamma$ values, our method exhibits a steeper improvement, overtaking all for $\gamma > 1$. Regarding recall, NBS and UC perform best at low $\gamma$ ($\approx 76\%$) but plateau as the signal grows, whereas our approach continues to improve, approaching UC near $\gamma \approx 1.9$. PPI remains limited to about 40\% recall due to its dependence on a single seed region, which constrains detection to a subset of affected clusters.

\begin{figure}[!htb]
    \centering

    \includegraphics[width = 0.8\linewidth]{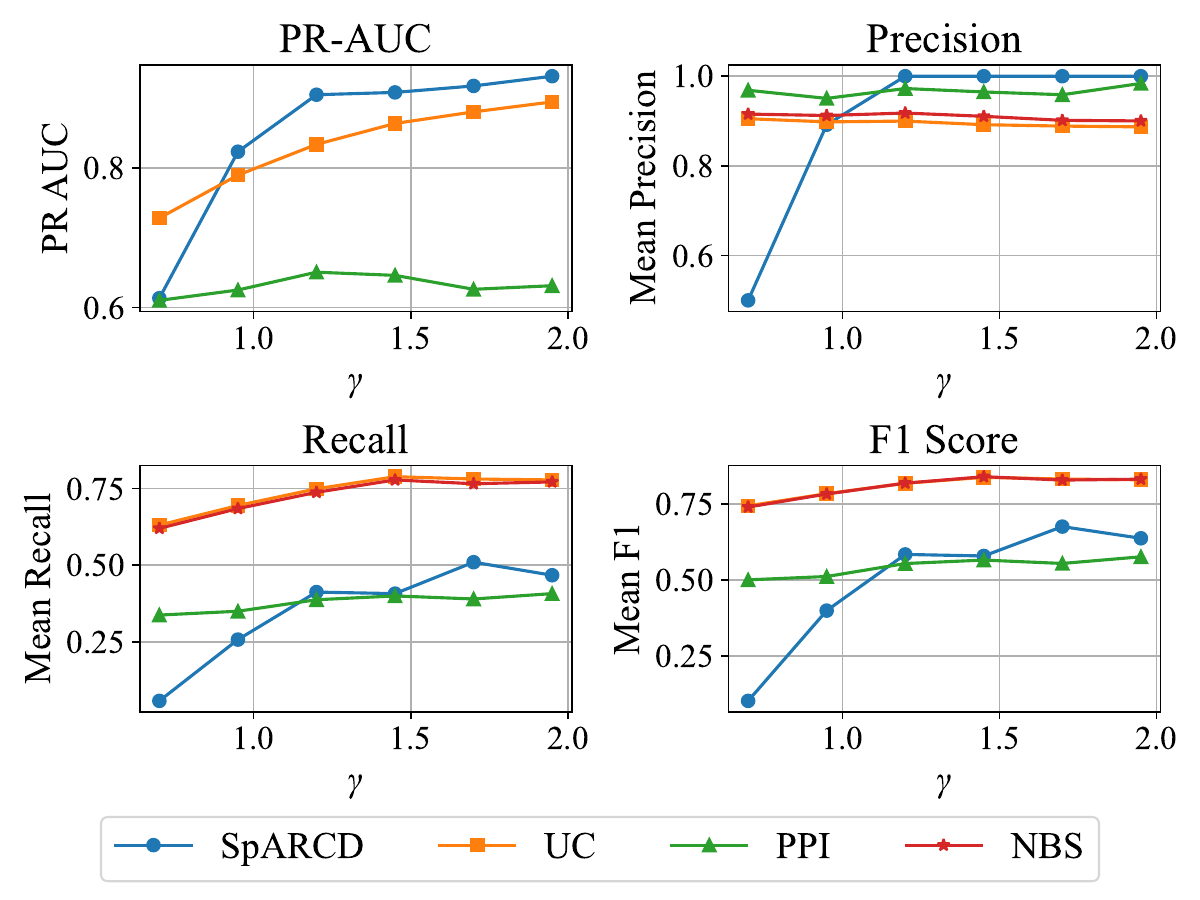}
    
    \caption{Linear simulation setting. The performance of SpARCD and other competing methods was evaluated in a linear simulation setting that utilized a block-diagonal covariance structure. The results, including the F1 score (bottom left), Precision (top left), Recall (bottom right), and PR–AUC (top right), are presented as functions of the signal strength parameter \(\gamma\). Higher values of \(\gamma\) indicate stronger and more structured linear dependencies among the clusters.}
    \label{fig:linear simulation results}
\end{figure}

\paragraph{Nonlinear Setting}

We first generated, for each sample \( i \), a set of \( d \) independent seed signals of length \( T \), denoted \( S_X^{(i)} \in \mathbb{R}^{d \times T} \). Each seed signal was drawn from a standard multivariate normal distribution
\begin{equation}\label{eq:seed}
 S_X^{(i)}(k,t) \sim \mathcal{N}(0, 1), \quad \text{for } k = 1, \ldots, q \quad \text{and } t = 1,\ldots,T \, .
\end{equation}
For $X$, we used a fixed block size, denoted $l$. The seed signals were then embedded as the first feature in each block via
\[
X^{(i)}_{(\text{nonlin})}(1 + B (k-1),t) = S_X^{(i)}(k,t) \quad k = 1,\ldots,d \, .
\]
The remaining features within each block ($r=1,\ldots,l-1$) were generated as a nonlinear transformation of the corresponding seed signal using sine functions with random phases $\phi_r$ and frequencies $f_r$ such that,
\[
X^{(i)}_{(\text{nonlin})}(1+ l  (k-1) + r,t) = \sin\left( \pi  f_{r}  S_x^{(i)}(k,t) + \phi_{r} \right) + \epsilon, \quad   r = 1,\ldots,(l-1), \quad k = 1,\ldots,q
\]
where $\epsilon \sim \mathcal N(0,\sigma^2 I)$ and $\sigma^2$ controls the noise level.
As in the linear setting, the samples $Y_{(\text{nonlin})}^{(i)}$ were generated similarly to $X_{(\text{nonlin})}^{(i)}$, with a different block structure such that the last block in $X$ is partitioned into two blocks in $Y$.

\textbf{Parameters for nonlinear setting:} The number of blocks in $X$ and $Y$ is set to $q = 8$ and $q+1=9$, respectively, with $l=18$ dependent features per block. 

\textbf{Full results:}
As in the linear case, the top panel of Figure \ref{fig:nonlinear_simulation} displays $\widehat s(r)$ values against their permutation-based null distribution. The discriminative signal is considerably clearer here: our test statistic highlights the differentiating regions with minimal background noise. The corresponding distance-correlation matrices further confirm that detected regions coincide with genuine connectivity changes between $X$ and $Y$.
 
Figure \ref{fig:nonlinear_simulation_results} summarizes the quantitative comparison. As noise levels ($\sigma$) increase, PPI exhibits moderate PR-AUC based on its raw outputs but yields no significant detections after multiple-testing correction, resulting in zero recall and precision across all $\sigma$. UC and NBS also fail to detect any signal. In contrast, our approach maintains perfect precision and recall up to $\sigma \approx 0.7$ and remains superior at higher noise levels, though with some expected performance degradation as the task becomes more challenging.

\begin{figure}[!htb]
\centering
\includegraphics[width=1\linewidth]{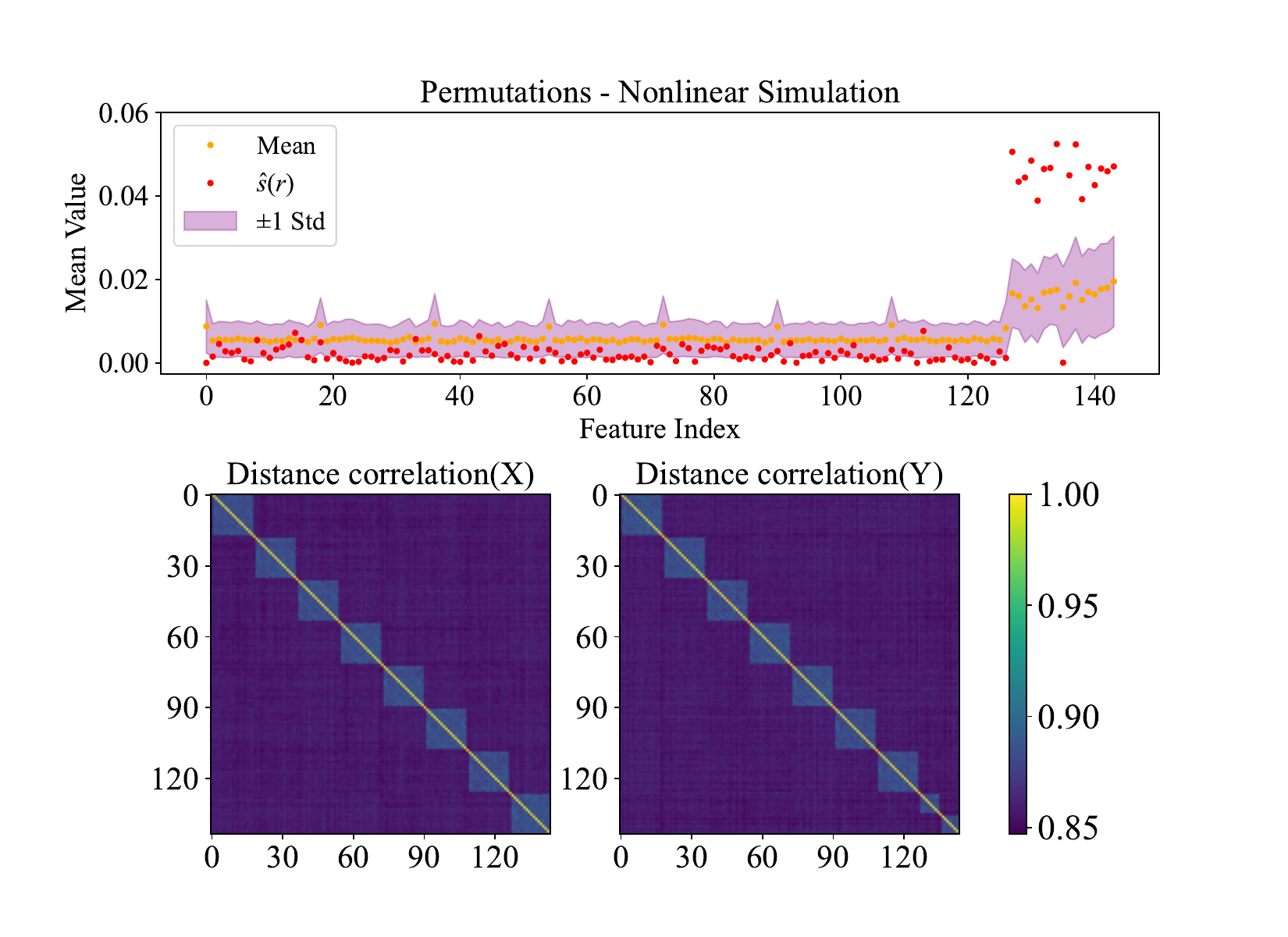}

\caption{Nonlinear simulation setting.
\textbf{Top panel:} Observed test statistic $\widehat s(r)$ (red) compared with the mean (orange) and standard deviation (purple) of the permutation-based null distribution.
\textbf{Bottom panel:} Distance-correlation matrices for datasets $X$ and $Y$, demonstrating that the detected regions correspond to genuine nonlinear differences in connectivity.}
\label{fig:nonlinear_simulation}
\end{figure}

\begin{figure}[!htb]
\centering

\includegraphics[width =0.8\linewidth]{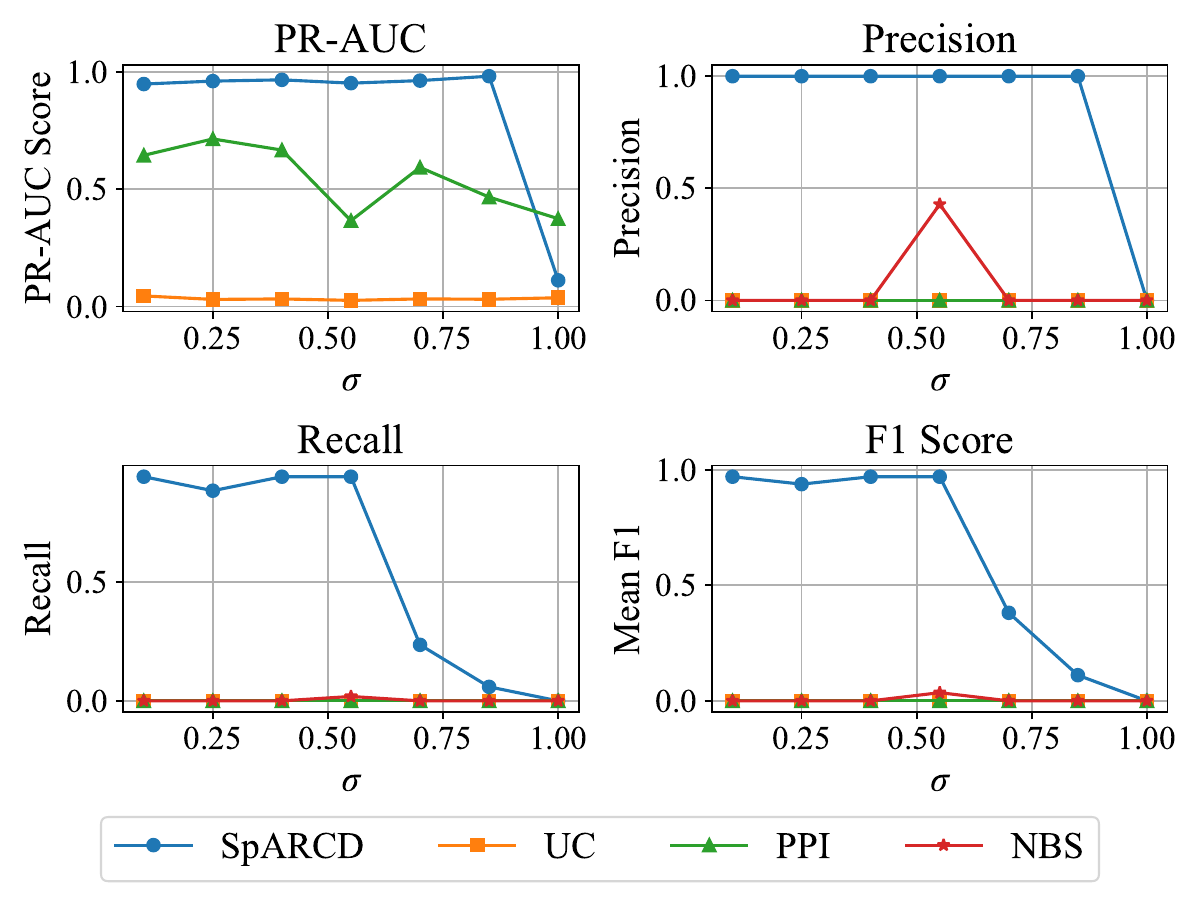}

\caption{Nonlinear simulation setting.
Performance of SpARCD and competing methods in the nonlinear simulation setting under varying noise levels ($\sigma$).
The results, including the F1 score (bottom left), Precision (top left), Recall (bottom right), and PR–AUC (top right), are presented as functions of the noise level. The results show that SpARCD maintains high accuracy and robustness even as noise increases, whereas competing methods rapidly lose power after multiple-testing correction.}
\label{fig:nonlinear_simulation_results}
\end{figure}

\paragraph*{Hybrid Setting}
\textbf{Parameters for hybrid setting:} The number of blocks in $X$ and $Y$ is set to $q = 8$ and $q+1=9$, respectively. The parameter $\gamma$, which determines the level of dependency in the linear setting, is set to $1.5$.

\textbf{Full results:}
Figure~\ref{fig: recall and precision - hybrid A} shows that NBS and UC gradually improve and slightly exceed our recall for $\alpha \geq 0.8$. Nonetheless, our method retains high precision and strong overall performance across all regimes, reflecting its robustness to varying dependency structures. PPI remains consistently inferior, highlighting the limitation of seed-based inference in detecting distributed connectivity changes.

\begin{figure}[htb!]
    \centering
    \includegraphics[width = 0.8\linewidth]{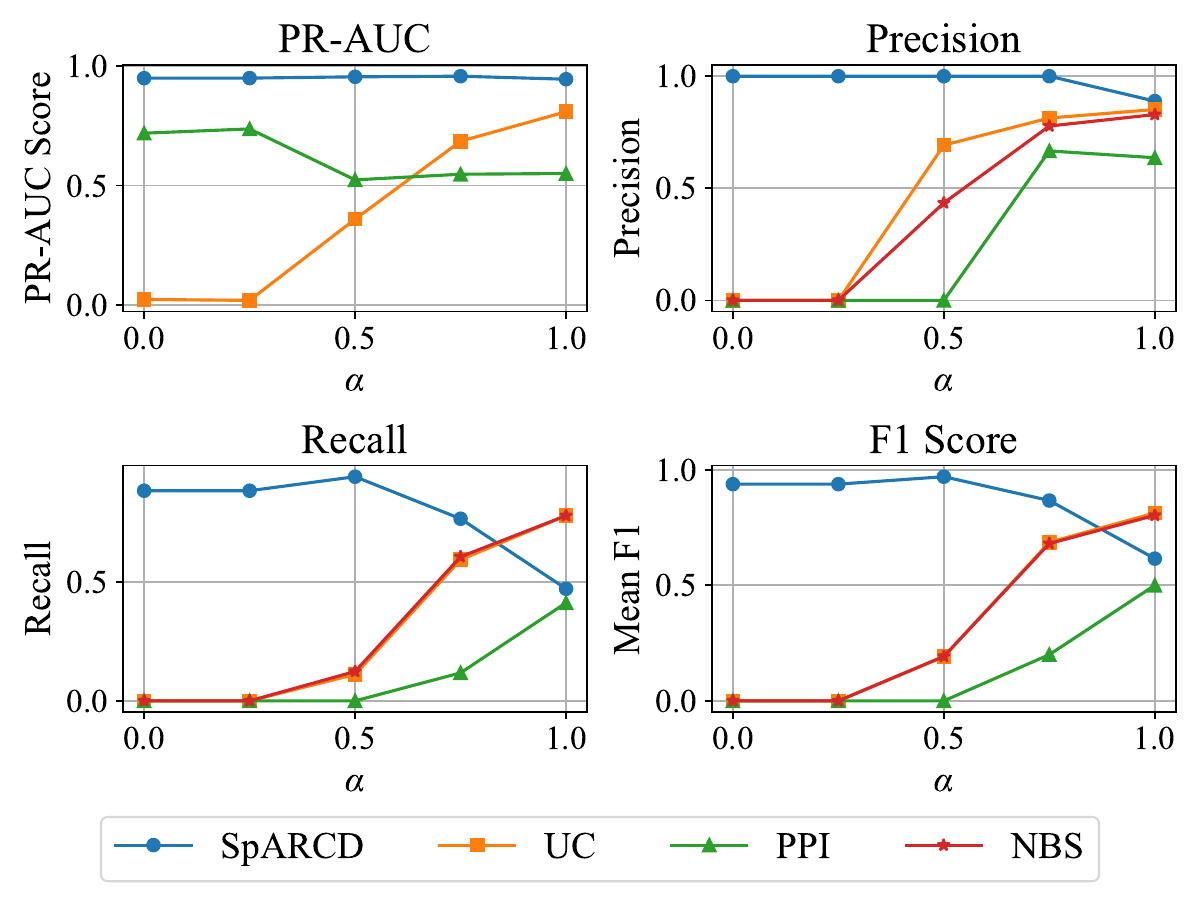}
    \caption{Performance of SpARCD and competing methods in the hybrid simulation setting with varying degrees of linearity ($\alpha$).
    The results, including the F1 score (bottom left), Precision (top left), Recall (bottom right), and PR–AUC (top right), are presented as functions of the mixing parameter $\alpha$, where smaller values indicate stronger nonlinear effects. SpARCD achieves the best overall performance for low-to-moderate $\alpha$, while maintaining competitive accuracy as dependencies become predominantly linear.}
    \label{fig: recall and precision - hybrid A}
\end{figure}

\paragraph{Edge-Wise Perturbation}
A baseline covariance matrix $\Sigma_A \in \mathbb{R}^{R \times R}$ is constructed from a single empirical sample by computing the Pearson correlation matrix of one randomly selected scan. Next, a subset of edges, denoted as $\mathcal{E} \subset \{(r,r'): r < r'\}$, is chosen uniformly at random. The size of this subset is determined by a specified sparsity parameter. The perturbed covariance matrix $\Sigma_B$ is obtained by increasing the corresponding entries,
\[
(\Sigma_B)_{r,r'} = (\Sigma_A)_{r,r'} + \delta, \quad \text{for } (r,r') \in \mathcal{E},
\]
followed by a projection onto the positive semidefinite cone to ensure a valid covariance structure.

Since the PSD projection and subsequent normalization may induce additional indirect changes in the connectivity structure, the ground truth is defined based on the resulting correlation matrices. Specifically, letting $\mathrm{Corr}_A$ and $\mathrm{Corr}_B$ denote the correlation matrices corresponding to $\Sigma_A$ and $\Sigma_B$, respectively, we define the ground truth as
\[
G = \left\{ (r,r') : |\mathrm{Corr}_B(r,r') - \mathrm{Corr}_A(r,r')| > \tau \right\},
\]
where $\tau$ is set to a data-driven threshold based on the empirical distribution of correlation differences.

\textbf{Parameters for edge-wise perturbation setting:}
We set the sparsity parameter to $0.1\%$, which defines the number of edges selected from the unique edges in the covariance matrix. Thus, $|\mathcal E| = 0.001 \times p(p-1)/2$. 

\textbf{Full results:}

Figure \ref{fig: recall and precision - edge perturbation} shows that the NBS and the UC precision decline as $\delta$ increases, reflecting a rise in false positives with stronger perturbations. PPI has moderate recall but very low precision. Our method achieves high precision for $\delta > 0.8$ and outperforms both NBS and UC in this regime, but recall remains low across all values of $\delta$.

\begin{figure}[htb!]
    \centering
    \includegraphics[width = 0.8\linewidth]{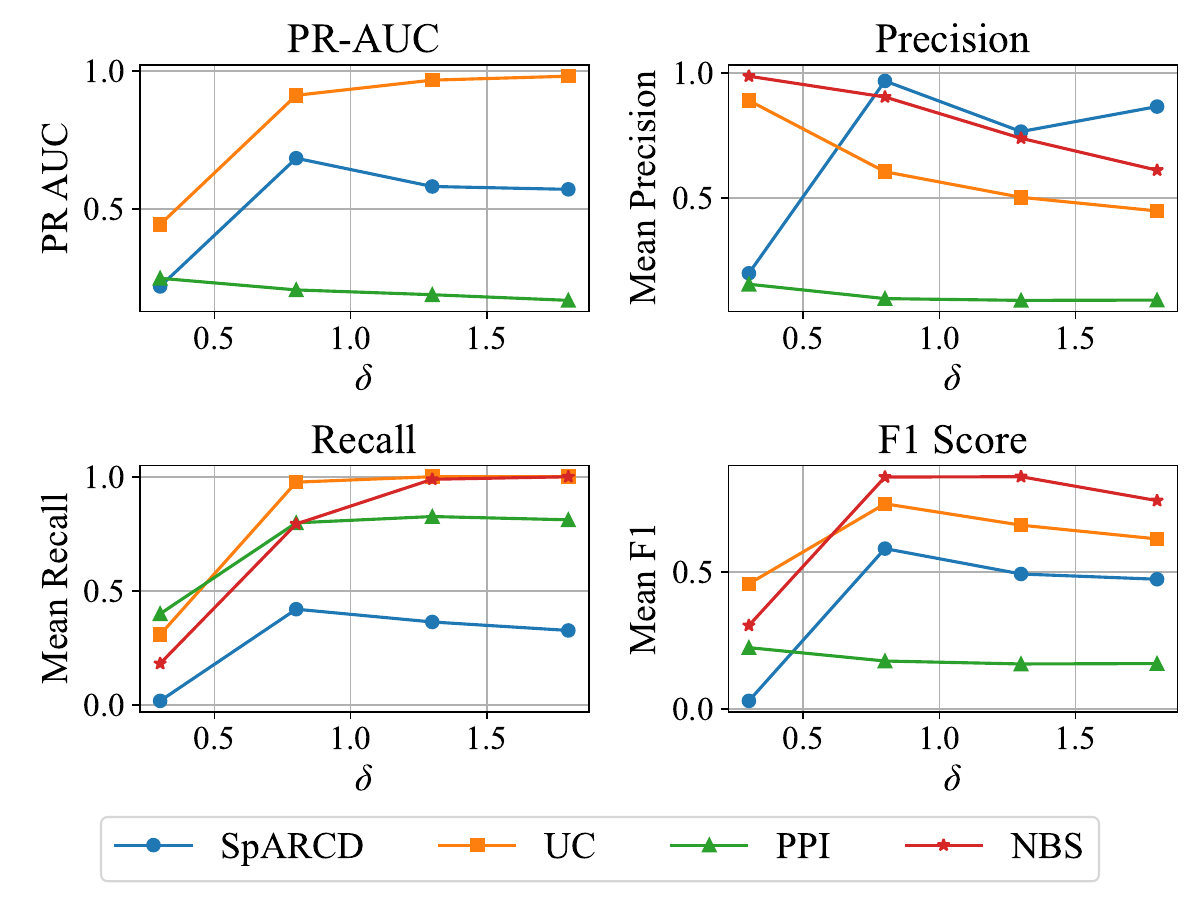}
    \caption{Performance of SpARCD and competing methods in the edge perturbation setting across different degrees of perturbation strength ($\delta$).
The results, including the F1 score (bottom left), Precision (top left), Recall (bottom right), and PR–AUC (top right), are presented as functions of the strength parameter $\delta$, where smaller values indicate weaker differences between the states.
SpARCD achieves high precision but lower recall than other methods.}
    \label{fig: recall and precision - edge perturbation}
\end{figure}

\paragraph*{Node-Wise Perturbation}

Let $W \in \mathbb{R}^{R \times R}$ be a baseline connectivity matrix, obtained by the Pearson correlation matrix of a randomly selected empirical scan. We define a normalized matrix $\tilde{W} = W / \max{\lambda(W)}$, where $\max{\lambda(W)}$ is the largest eigenvalue of W. 

We define a stable VAR(1) process:
\[
x_t = A x_{t-1} + \varepsilon_t, \quad t = 1, \ldots, T,
\]
with the transition matrix $A = \beta \tilde{W}$ where $\beta < 1$ guarantees stability (i.e., all eigenvalues of A lie inside the unit circle). The noise process $\varepsilon_t \sim \mathcal{N}(0, \sigma^2 I )$ introduces stochastic variability among regions. 

To introduce node-wise perturbations, we select a subset of nodes $\mathcal{V} \subset \{1,\ldots,R\}$ and define a mask matrix $M$ such that $M_{r,r'} = 1$ if either $r \in \mathcal{V}$ or $r' \in \mathcal{V}$, and $0$ otherwise. The perturbed covariance matrix is defined as
\[
\Sigma_B = \frac{1}{\delta} \Sigma_A + \delta M,
\]
which amplifies connections involving nodes in $\mathcal{V}$ while attenuating the remaining connections. Time series for state $B$ are then generated using the same VAR model with the perturbed covariance.

The ground truth corresponds to the set of perturbed nodes $\mathcal{V}$, and evaluation is based on the resulting differences in connectivity between the two states.
\textbf{Parameters for node-wise perturbation setting:}
The parameters $\beta$ and $\sigma^2$ are set to $0.95$ and $0.1$ respectively. A total of $10$ nodes have been selected, so $|\mathcal V| = 10$.

\textbf{Full results:}
Figure \ref{fig: recall and precision - node perturbation} demonstrates that our method is achieving both high precision and high recall under this setting. NBS, followed by UC, also attains high recall, slightly higher than ours, but with lower precision. Notably, their precision initially increases with $\delta$ but eventually declines as false positives accumulate at higher perturbation levels. PPI shows competitive PR-AUC based on its raw outputs, and a similar recall to our method. However, it produces very low precision due to a high number of false detections.

\begin{figure}[tb]
    \centering
    \includegraphics[width = 0.8\linewidth]{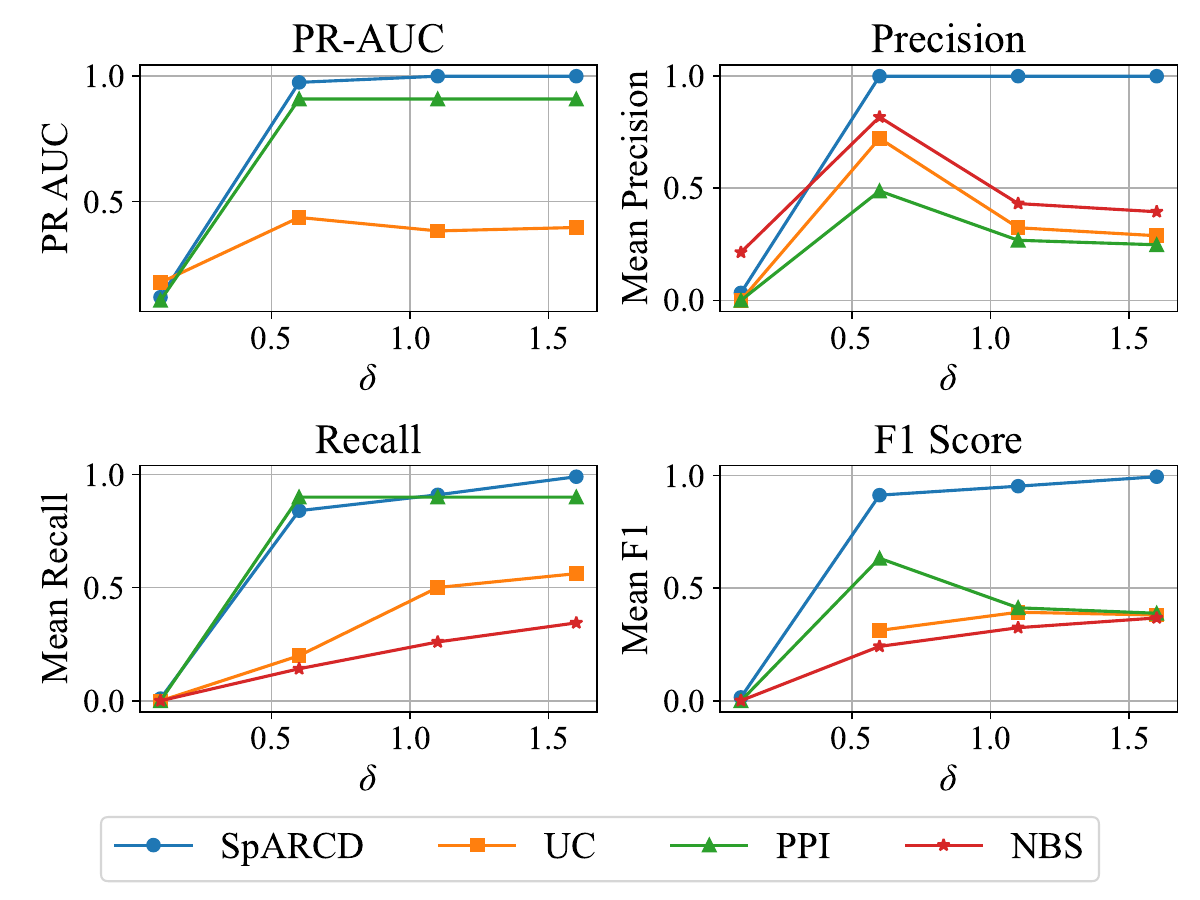}
    \caption{Performance of SpARCD and competing methods in the node perturbation setting across different degrees of perturbation strength ($\delta$).
The results, including the F1 score (bottom left), Precision (top left), Recall (bottom right), and PR–AUC (top right), are presented as functions of the strength parameter $\delta$, where smaller values indicate weaker differences between the states.
SpARCD achieves high precision and recall.}
    \label{fig: recall and precision - node perturbation}
\end{figure}


\subsection{Evaluation Metrics}

Precision is defined as the proportion of detected regions that are truly differentiating, 
\[
\text{Precision} = \frac{\text{true positives}}{\text{number of detections}} \, .
\]
Recall measures the proportion of truly differentiating regions that are successfully detected,
\[
\text{Recall} = \frac{\text{true positives}}{\text{number of ground truth regions}} \, .
\]
The F1 score is the harmonic mean of precision and recall, providing a single measure that balances the trade-off between them.
\[
F_1 = 2 \cdot \frac{\text{Precision} \cdot \text{Recall}}{\text{Precision} + \text{Recall}} \, .
\]
For each method, we compute a permutation-based p-value for every region and adjust for multiple comparisons using the BH procedure. Regions with BH-adjusted p-values below 0.05 are treated as positive detections, and precision and recall are computed accordingly. Further details on the computation of p-values for the PPI and UC methods are given in the supplementary materials.

To compute PR-AUC, regions are ranked by their continuous scores. For the UC method, we denote by $\widehat s_{\mathrm{UC}}(r,r')$ the estimated mean Pearson correlation between regions $r,r'$. In the case of PPI, a summary measure for each region is computed based on its interaction with a predefined seed region \citep{friston1997psychophysiological}. We denote the resulting summary measure for region $r$ by $\widehat s_{\mathrm{PPI}}(r)$.
For each threshold, we calculate the corresponding precision and recall, and the PR-AUC is the area under the precision-recall curve. This provides a threshold-independent summary of detection performance. For our approach and PPI, precision, recall, and PR-AUC are computed using $\widehat s(r)$ and $\widehat s_{PPI}(r)$, respectively. For the UC analysis, where $\widehat s_{UC}$ is a symmetric matrix, we use the elements of its vectorized upper-triangular part as a score for region pairs. Since NBS produces binary detections at the subnetwork level rather than continuous scores, its performance is evaluated using precision and recall only, and not via precision--recall curves.

\section{Additional Experiments on fMRI Connectivity}\label{app:additional experiments}

In addition to the primary EFMT analysis in PTSD patients (Section~\ref{sec: application to hariri}), we conducted two supplementary experiments to evaluate the robustness and generalizability of our method on EFMT data. In all cases, we compute connectivity graphs for the two states $A$ and $B$, apply the spectral differential operator, and calculate the regional test statistic $s(r)$ as described in the main text.

\paragraph*{Emotional vs. Neutral Conditions in Control Group} \label{subsec: addition efmt - non ptsd}
We applied the proposed SpARCD framework to EFMT data from the subset of control participants in the cohort described in Section~\ref{sec: problem setting} ($n = 47$). The preprocessing, ROI parcellation, and block structure are identical to those described for the PTSD cohort in Section~\ref{sec: application to hariri}. For each participant, we constructed paired datasets $X, Y \in \mathbb{R}^{n \times 113 \times 68}$ corresponding to the face ($A$) and shape ($B$) blocks of the EFMT. We then computed the regional test statistic $\widehat s(r)$, and significance was assessed via a permutation test ($B=2000$) as described in Step IV Section \ref{sec. sparcd} with BH correction at $\alpha=0.05$. 
The results were broadly consistent with those obtained for the PTSD cohort. Regions \{44, 45,76, 77, 78, 79, 94, 95\} exhibited significant connectivity differences, overlapping substantially with those identified in PTSD participants. In contrast, the regions \{ 62, 63\} had a low score in the control group, indicating that there is no differentiating connectivity for these groups.\\ 
Overall, these findings indicate that the connectivity differences between face and shape blocks are largely similar between the PTSD and control groups, suggesting that the pattern of task-evoked connectivity is largely preserved across these groups.  
\begin{figure}[tb]
\centering
\includegraphics[width=0.95\linewidth]{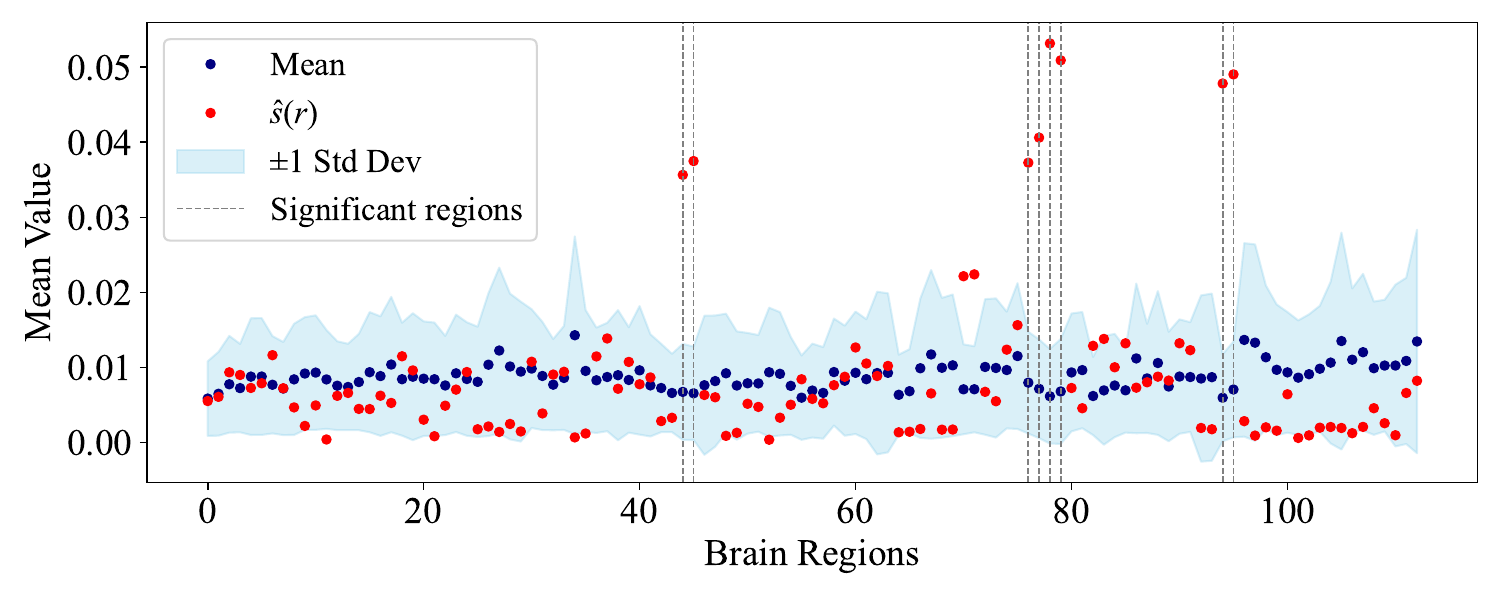}
\caption{Results of the SpARCD analysis applied to the Hariri fMRI dataset for the non-PTSD patient cohort.
The observed test statistic $\widehat s(r)$ (red) is shown together with the mean (blue) and standard deviation (light blue) of the permutation-based null distribution. Several regions exhibit clear deviations above the null expectation, indicating significant differences in functional connectivity between the emotional and neutral task conditions.}
\end{figure}
\paragraph*{EFMT vs. Resting-State fMRI in PTSD Subjects} \label{subsec: addition efmt - rfmri vs efmt}
We compared task-based connectivity during the EFMT to resting-state connectivity within the PTSD cohort described in Section~\ref{sec: problem setting} ($n = 113$). The preprocessing and ROI parcellation were identical to the main analysis. To match time series lengths, the edges of the resting-state scans were truncated. For each participant, we constructed paired datasets $X, Y \in \mathbb{R}^{n \times 113 \times T}$ corresponding to EFMT ($X$) and resting-state ($Y$) signals. We then computed the test statistic $\widehat s(r)$, with statistical significance evaluated by a standard permutation test ($B = 5000$) followed by BH correction at $\alpha = 0.05$. \\
In this analysis, most of the identified regions are subsets of those found in the comparison between emotional and neutral stimuli. Regions \{62, 63, 77, 78, 79, 92, 93, 94, 95\} showed significant differences in connectivity. Notably, regions \{92, 93\} were previously unidentified. In contrast, regions \{44, 45\} did not achieve statistical significance in this comparison.\\
A visual inspection of the $\widehat s(r)$ profiles reveals that the elevated $\widehat s(r)$ values are still apparent in the unidentified regions when compared to the permuted distributions. This suggests that the discrepancies are likely due to a higher variance in the permuted statistic, $\widehat{s}^{(\pi)}(r)$. This higher variance may occur because we are permuting between two distinctly different scan types.
\begin{figure}[tb]
\centering
\includegraphics[width=0.95\linewidth]{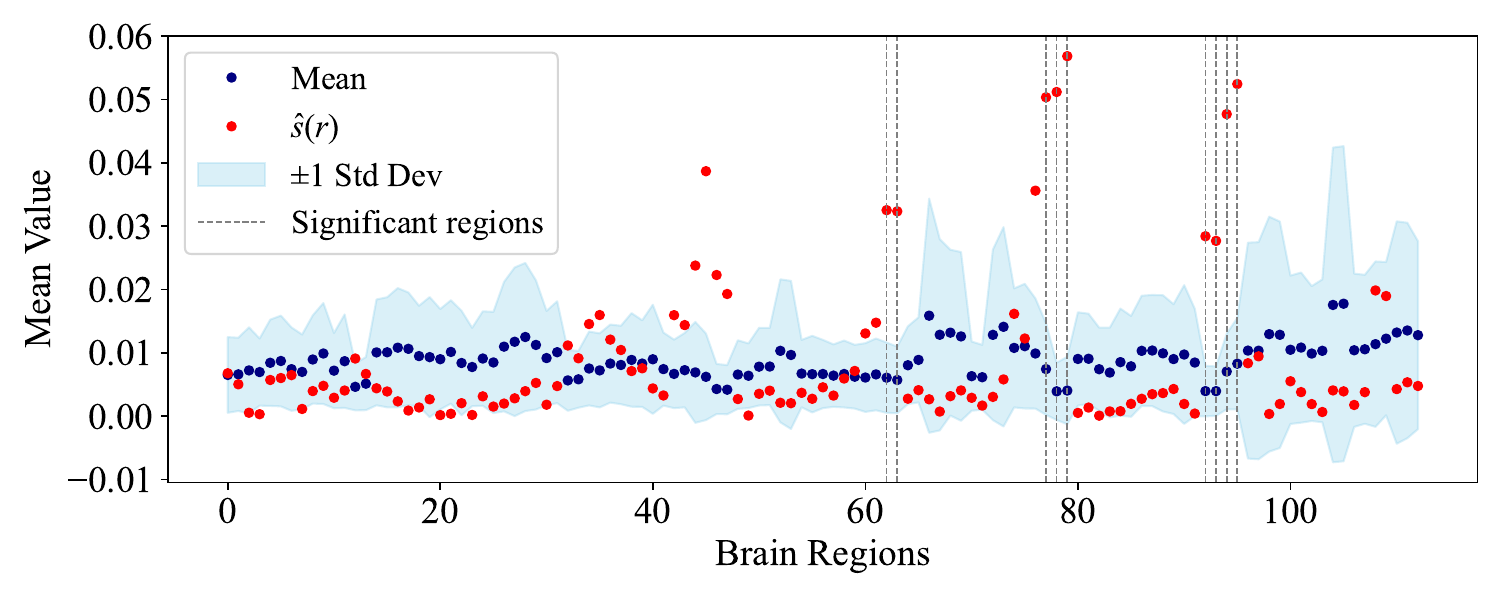 }
\caption{Results of SpARCD for the comparison of the Hariri fMRI dataset and resting-state fMRI on PTSD patients.\\
The observed test statistic $\widehat s(r)$ (red) is shown together with the mean (blue) and standard deviation (light blue) of the permutation-based null distribution. Several regions exhibit clear deviations above the null expectation, indicating significant differences in functional connectivity between EFMT and resting state conditions.}
\end{figure}

\section{Atlas-ROI mapping}\label{app:Atlas mapping}
A mapping of ROI numbers to their corresponding anatomical names.
\begin{center}
\begin{longtable}{|c|l|}
\caption{Mapping between ROI indices and brain region names (Harvard-Oxford Atlas).} \label{tab:harvard_roi_mapping} \\
\hline
ROI Index & Brain Region \\
\hline
\endfirsthead
\caption[]{Mapping between ROI indices and brain region names (Harvard-Oxford Atlas).} \\
\hline
ROI Index & Brain Region \\
\hline
\endhead
\hline
\multicolumn{2}{r}{Continued on next page} \\
\hline
\endfoot
\hline
\endlastfoot
0 & Left Frontal Pole \\
1 & Right Frontal Pole \\
2 & Left Insular Cortex \\
3 & Right Insular Cortex \\
4 & Left Superior Frontal Gyrus \\
5 & Right Superior Frontal Gyrus \\
6 & Left Middle Frontal Gyrus \\
7 & Right Middle Frontal Gyrus \\
8 & Left Inferior Frontal Gyrus pars triangularis \\
9 & Right Inferior Frontal Gyrus pars triangularis \\
10 & Left Inferior Frontal Gyrus pars opercularis \\
11 & Right Inferior Frontal Gyrus pars opercularis \\
12 & Left Right Precentral Gyrus \\
13 & Right Precentral Gyrus \\
14 & Left Temporal Pole \\
15 & Right Temporal Pole \\
16 & Left Superior Temporal Gyrus anterior division \\
17 & Right Superior Temporal Gyrus anterior division \\
18 & Left Superior Temporal Gyrus posterior division \\
19 & Right Superior Temporal Gyrus posterior division \\
20 & Left Middle Temporal Gyrus anterior division \\
21 & Right Middle Temporal Gyrus anterior division \\
22 & Left Middle Temporal Gyrus posterior division \\
23 & Right Middle Temporal Gyrus posterior division \\
24 & Left Middle Temporal Gyrus temporooccipital part \\
25 & Right Middle Temporal Gyrus temporooccipital part \\
26 & Left Inferior Temporal Gyrus anterior division \\
27 & Right Inferior Temporal Gyrus anterior division \\
28 & Left Inferior Temporal Gyrus posterior division \\
29 & Right Inferior Temporal Gyrus posterior division \\
30 & Left Inferior Temporal Gyrus temporooccipital part \\
31 & Right Inferior Temporal Gyrus temporooccipital part \\
32 & Left Postcentral Gyrus \\
33 & Right Postcentral Gyrus \\
34 & Left Superior Parietal Lobule \\
35 & Right Superior Parietal Lobule \\
36 & Left Supramarginal Gyrus anterior division \\
37 & Right Supramarginal Gyrus anterior division \\
38 & Left Supramarginal Gyrus posterior division \\
39 & Right Supramarginal Gyrus posterior division \\
40 & Left Angular Gyrus \\
41 & Right Angular Gyrus \\
42 & Left Lateral Occipital Cortex superior division \\
43 & Right Lateral Occipital Cortex superior division \\
\hlroi{44} & \hlroi{Left Lateral Occipital Cortex inferior division} \\
\hlroi{45} & \hlroi{Right Lateral Occipital Cortex inferior division} \\
46 & Left Intracalcarine Cortex \\
47 & Right Intracalcarine Cortex \\
48 & Left Frontal Medial Cortex \\
49 & Right Frontal Medial Cortex \\
50 & Left Juxtapositional Lobule Cortex (formerly Supplementary Motor Cortex) \\
51 & Right Juxtapositional Lobule Cortex (formerly Supplementary Motor Cortex) \\
52 & Left Subcallosal Cortex \\
53 & Right Subcallosal Cortex \\
54 & Left Paracingulate Gyrus \\
55 & Right Paracingulate Gyrus \\
56 & Left Cingulate Gyrus anterior division \\
57 & Right Cingulate Gyrus anterior division \\
58 & Left Cingulate Gyrus posterior division \\
59 & Right Cingulate Gyrus posterior division \\
60 & Left Precuneous Cortex \\
61 & Right Precuneous Cortex \\
\hlroi{62} & \hlroi{Left Cuneal Cortex} \\
\hlroi{63} & \hlroi{Right Cuneal Cortex} \\
64 & Left Frontal Orbital Cortex \\
65 & Right Frontal Orbital Cortex \\
66 & Left Parahippocampal Gyrus anterior division \\
67 & Right Parahippocampal Gyrus anterior division \\
68 & Left Parahippocampal Gyrus posterior division \\
69 & Right Parahippocampal Gyrus posterior division \\
70 & Left Lingual Gyrus \\
71 & Right Lingual Gyrus \\
72 & Left Temporal Fusiform Cortex anterior division \\
73 & Right Temporal Fusiform Cortex anterior division \\
74 & Left Temporal Fusiform Cortex posterior division \\
75 & Right Temporal Fusiform Cortex posterior division \\
\hlroi{76} & \hlroi{Left Temporal Occipital Fusiform Cortex} \\
\hlroi{77} & \hlroi{Right Temporal Occipital Fusiform Cortex} \\
\hlroi{78} & \hlroi{Left Occipital Fusiform Gyrus} \\
\hlroi{79} & \hlroi{Right Occipital Fusiform Gyrus} \\
80 & Left Frontal Operculum Cortex \\
81 & Right Frontal Operculum Cortex \\
82 & Left Central Opercular Cortex \\
83 & Right Central Opercular Cortex \\
84 & Left Parietal Operculum Cortex \\
85 & Right Parietal Operculum Cortex \\
86 & Left Planum Polare \\
87 & Right Planum Polare \\
88 & Left Heschl's Gyrus (includes H1 and H2) \\
89 & Right Heschl's Gyrus (includes H1 and H2) \\
90 & Left Planum Temporale \\
91 & Right Planum Temporale \\
92 & Left Supracalcarine Cortex \\
93 & Right Supracalcarine Cortex \\
\hlroi{94} & \hlroi{Left Occipital Pole} \\
\hlroi{95} & \hlroi{Right Occipital Pole} \\
96 & Left Lateral Ventricle \\
97 & Right Lateral Ventricle \\
98 & Left Thalamus \\
99 & Right Thalamus \\
100 & Left Caudate \\
101 & Right Caudate \\
102 & Left Putamen \\
103 & Right Putamen \\
104 & Left Pallidum \\
105 & Right Pallidum \\
106 & Left Hippocampus \\
107 & Right Hippocampus \\
108 & Left Amygdala \\
109 & Right Amygdala \\
110 & Left Accumbens \\
111 & Right Accumbens \\
112 & Brainstem \\
\end{longtable}
\end{center}

\bibliographystyle{elsarticle-harv} 
\bibliography{reference}






\end{document}